\documentclass[12 pt]{article}
\usepackage{bbm}
\usepackage[longnamesfirst]{natbib}
\usepackage{amsmath}
\usepackage{amssymb}
\usepackage{amsthm}
\usepackage{setspace}
\usepackage[colorlinks]{hyperref}
\usepackage{graphicx}
\usepackage{caption}
\usepackage{enumitem}
\setlist{noitemsep,parsep=6pt,partopsep=0pt,topsep=0pt}
\usepackage[margin=1in]{geometry}
\usepackage{verbatim}
\usepackage{sectsty}
\usepackage[svgnames]{xcolor}
\usepackage{subfiles}
\usepackage{xargs}
\usepackage[titletoc,title]{appendix}
\usepackage{titlesec}
\usepackage[bottom]{footmisc}
\usepackage[multiple]{fnpct}
\usepackage[colorinlistoftodos,textsize=tiny]{todonotes}
\usepackage[capitalize,nameinlink]{cleveref}
\usepackage[normalem]{ulem}
\usepackage[font={small,it},justification=justified]{caption}
\usepackage{tablefootnote}
\usepackage{threeparttable}
\usepackage{afterpage}
\usepackage{xparse}
\hypersetup{
   pdftitle={Replacement and Reputation},
    pdfauthor={Navin Kartik, Elliot Lipnowski, Harry Pei},
    citecolor=DarkBlue,
    bookmarksnumbered=true,
    urlcolor=Indigo,
    linkcolor=DarkBlue,
}
\usepackage{subcaption}
\usepackage{xspace}
\usepackage{booktabs}

\theoremstyle{plain}
\numberwithin{equation}{section}

\newtheorem{Proposition}{Proposition}
\newtheorem{Theorem}{Theorem}
\newtheorem{Lemma}{Lemma}

\newtheorem{Corollary}{Corollary}

\newtheorem*{FEI}{Condition FEI (\emph{Full Effort Incentives})}

\newtheorem*{FBI*}{Condition FBI*}

\newtheorem{notation}{Notation}

\renewcommand{\epsilon}{\varepsilon}

\newcommand{\condref}[1]{\hyperref[#1]{Condition FEI}\xspace} 
\NewDocumentCommand{\citepos}{sm}{%
  \IfBooleanTF{#1}
    {\citeauthor*{#2}'s~\citeyearpar{#2}}  
    {\citeauthor{#2}'s~\citeyearpar{#2}}   
}
\newcommand{\R}{\mathbb{R}}
\newcommand{\Rpos}{\R_{\geq0}}
\newcommand{\Z}{\mathbb{Z}}
\newcommand{\Zpos}{\Z_{\geq0}}
\newcommand{\Zstrpos}{\Z_{>0}}
\newcommand{\N}{\Zstrpos}
\newcommand{\E}{\mathbb{E}}
\newcommand{\Prob}{\mathbb{P}}

\newcommand{\Act}{A} 
\newcommand{\heff}{1} 
\newcommand{\leff}{0} 
\newcommand{\mon}{f}

\newcommand{\upo}{u_P} 
\newcommand{\uvo}{u_V} 
\newcommand{\cost}{c} 
\newcommand{\ecost}{\kappa} 
\newcommand{\vdiscount}{\delta_V} 
\newcommand{\survprob}{\rho} 
\newcommand{\aprob}{\tilde a} 
\newcommand{\belief}{\pi}
\newcommand{\bupdate}{\beta}
\newcommand{\pbound}{\eta}
\newcommand{\lrat}{\lambda}
\newcommand{\bbound}{\bar\bupdate}
\newcommand{\ooption}{u_0}
\newcommand{\prior}{\belief_0}
\newcommand{\bprior}{\bar\belief_0}
\newcommand{\Hist}{\mathcal{H}}
\newcommand{\pstrat}{\sigma_P}
\newcommand{\vstrat}{\sigma_V}

\newcommand{\brac}[1]{\left[#1\right]}
\newcommand{\paren}[1]{\left(#1\right)}
\newcommand{\curlyb}[1]{\left\{#1\right\}}

\newcommand{\actrv}{\mathbf{a}}
\newcommand{\beliefrv}{\boldsymbol{\belief}}
\newcommand{\replacerv}{\mathbf{r}}
\newcommand{\replacetimerv}{{\boldsymbol{\tau}_\replacerv}}
\newcommand{\histrv}{\mathbf{h}}
\newcommand{\actprv}{\tilde{\mathbf{a}}}
\newcommand{\eeff}{\alpha} 
\newcommand{\cval}{v} 

\newcommand{\lemnode}[1]{\hyperref[#1]{L\ref*{#1}}}
\newcommand{\propnode}[1]{\hyperref[#1]{P\ref*{#1}}}
\newcommand{\cornode}[1]{\hyperref[#1]{C\ref*{#1}}}
\newcommand{\thmnode}[1]{\hyperref[#1]{T\ref*{#1}}}

\usetikzlibrary{patterns, decorations.pathreplacing, positioning, arrows.meta,calc,decorations.pathreplacing,calligraphy,shapes.geometric,decorations.markings}

\let \savenumberline \numberline
\def \numberline#1{\savenumberline{#1.}}

\makeatletter
  \renewcommand\@seccntformat[1]{\csname the#1\endcsname.{\hskip.7em\relax}} 
\makeatother
\makeatletter
\renewenvironment{proof}[1][\proofname] {\par\pushQED{\qed}\normalfont\topsep6\p@\@plus6\p@\relax\trivlist\item[\hskip\labelsep\bfseries#1\@addpunct{.}]\ignorespaces}{\popQED\endtrivlist\@endpefalse}
\makeatother

\newcommand{\mailto}[1]{\href{mailto:#1}{\texttt{#1}}} 

\let\oldfootnote\footnote
\renewcommand\footnote[1]{\oldfootnote{\hspace{.5mm}#1}}

\titlespacing\section{0pt}{10pt plus 2pt minus 2pt}{4pt plus 2pt minus 2pt} 
\titlespacing\subsection{0pt}{6pt plus 2pt minus 2pt}{2pt plus 2pt minus 2pt} 
\titlespacing\subsubsection{0pt}{6pt plus 2pt minus 2pt}{0pt plus 2pt minus 2pt} 
\titlespacing{\paragraph}{%
  0pt}{
  0.5\baselineskip}{
  1em}

\newcommand{\appendixref}[1]{\hyperref[#1]{Appendix \ref{#1}}}
\newcommand{\appendixlink}{\hyperref[sec:appendix]{Appendix}}
\newcommand{\secref}[1]{\S\ref{#1}}

\usepackage{xcolor}

\sectionfont{\color{DarkRed}}
\subsectionfont{\color{DarkRed}}
\subsubsectionfont{\color{DarkRed}}

\begin{document}

\title{\color{DarkRed}Replacement and Reputation\thanks{We thank Richard Van Weelden for discussions that led to this project, as well as Nageeb Ali, Sandeep Baliga, Nina Bobkova, Alessandro Bonatti, Joyee Deb, Jeffrey Ely, Mehmet Ekmekci, Drew Fudenberg, Rohit Lamba, David Levine, Qingmin Liu,
Alessandro Pavan, Carlo Prato, Doron Ravid, \href{https://www.refine.ink}{Refine.ink}, Jo\~ao Ramos, Larry Samuelson, Andrzej Skrzypacz,
Bruno Strulovici, Mike Ting, Juuso V\"{a}lim\"{a}ki, 
Alexander Wolitzky, and various audiences for helpful comments. Maxim Ventura and Cole Wittbrodt provided excellent research assistance.
}}
\author{Navin Kartik\thanks{Department of Economics, Yale University.  Email: \mailto{navin.kartik@yale.edu}.} \and Elliot Lipnowski\footnote{Department of Economics, Yale University. Email: \mailto{elliot.lipnowski@yale.edu}.} \and Harry Pei\footnote{Department of Economics, Northwestern University. Email: \mailto{harrydp@northwestern.edu}.}}
\date{\today}

\maketitle
\thispagestyle{empty}

\onehalfspacing
\setlength{\parskip}{6pt plus 1pt minus 1pt} 

\begin{abstract}
Does electoral replacement ensure that officeholders eventually act in voters' interests? We study a reputational model of accountability. Voters observe incumbents' performance and decide whether to replace them. Politicians may be ``good'' types who always exert effort or opportunists who may shirk. We find that good long-run outcomes are always attainable, though the mechanism and its robustness depend on economic conditions. In environments conducive to incentive provision, some equilibria feature sustained effort, yet others exhibit some long-run shirking. In the complementary case, opportunists are never fully disciplined, but selection dominates: every equilibrium eventually settles on a good politician, yielding permanent effort.

\end{abstract}

\newpage
\setcounter{page}{1}

\section{Introduction}
\label{sec:intro}

In various delegated decision-making settings, a principal’s ultimate tool is the power to replace an agent.  A firm’s board can fire a poorly performing CEO; patients can switch health-care providers after bad service; and ineffective bureaucrats or organizational leaders may be dismissed. But the most prominent context is democratic politics, where the primary instrument voters use to control officeholders is re-election. An influential body of work, starting with \citet{Barro73} and \citet{Ferejohn86}, has studied the consequences of voters' authority to replace incumbents. Such authority is essential not only to discipline politicians' behavior, but also to select good politicians \citep{Fearon99,Besley05}.

Indeed, much further back, \citet{Mill1859} stressed the importance of the \emph{long-run} implications of selection, writing that ``The worth of a State, in the long run, is the worth of the individuals composing it.'' Yet we know surprisingly little from existing formal analyses about the {long-run} consequences of replacement in environments with both adverse selection and moral hazard. Much of the theoretical political-accountability literature (elaborated subsequently) analyzes models with short horizons or with short term limits.\footnote{This contrasts with an empirical literature on the considerable length of political careers. See, for example, structural estimates by \citet{DKM05} and a study on legislative effectiveness by \citet{PMS06}, both highlighting how performance and retention of officeholders evolve significantly over long tenures.} A fundamental question thus remains: absent other frictions, does the replacement mechanism ensure that eventually voters can identify good politicians, and/or that officeholders will act in voters' interests?

This paper answers that question in a simple and stylized model that we view as a natural benchmark. Our model takes inspiration from not only prior work on accountability, but also the {reputation} literature \citep[e.g.,][]{FL92}. While the model is applicable more broadly, for concreteness we filter it through a political lens.

A pool of long-lived politicians interacts with a sequence of short-lived (or myopic) voters. In each period, the current voter observes the incumbent's past performance and decides whether to retain him or replace him with a fresh draw from the politician pool. The officeholder---old or new---then chooses, covertly, whether to exert effort (work or shirk), which maps stochastically into performance.  Politicians are one of two types: a \emph{good} type who always exerts effort, or an \emph{opportunistic} type who values office but dislikes effort. Voters simply want effort;\footnote{They may also incur a cost of replacing incumbents, but we suppress that for this introduction.} but this means they instrumentally benefit from good types, and politicians have an incentive to build a reputation. Apart from the retention decision, voters have no recourse to any other incentive instruments such as transfers.

In this setting, we obtain a sharp but nuanced answer to the motivating question of long-run outcomes. On the one hand, there is always at least {one} equilibrium that eventually has sustained effort from incumbents.\footnote{As elaborated in \autoref{sec: model}, our equilibrium notion is that of ``personal symmetric'' weak perfect Bayesian equilibrium. Roughly speaking, ``personal'' means that the interaction between every incumbent and the voters is self-contained, and ``symmetric'' means it is ex-ante identical across all the politicians.} On the other hand, whether this good long-run outcome arises in \emph{every} equilibrium depends on the economic environment. Intuitively, the environment is conducive to incentive provision if effort costs are low, monitoring (i.e., how performance reflects effort) is precise, and politicians are patient. Then there is an equilibrium in which even opportunistic incumbents exert effort in every period (even in the short run). But we show that in this case, there also exist equilibria in which, even in the long run, opportunistic politicians will periodically hold office and shirk. We identify a novel feature of all such equilibria that do not deliver good long-run outcomes: \emph{replacement despite a favorable reputation}. That is, voters sometimes replace an incumbent even when his past performance indicates that he is more likely to be a good type than the average politician. Voters do so when they (correctly) anticipate that an opportunistic incumbent with the favorable reputation is sufficiently more likely to shirk than a newly-installed opportunist. The upshot is that in environments conducive to incentive provision, equilibrium selection---which can be viewed as political norms---crucially matter for outcomes, even in the long run.

Now consider the complementary case when parameters are not conducive to incentive provision: there is no equilibrium that sustains effort in every period from opportunists. Paradoxically perhaps, this turns out to guarantee effort in the long run. The mechanism operates through selection.  We show that now, in every equilibrium, every opportunistic officeholder is eventually replaced, whereas eventually some good-type incumbent will be retained forever. The logic assuring long-run selection is subtle. To appreciate why, suppose good types are scarce. We establish that the value of a newly-installed politician (a voter's endogenous ``outside option'') is then \emph{arbitrarily low} across all equilibria, meaning that newly-installed officeholders almost certainly shirk. That makes voters reluctant to replace any incumbent. However, it turns out that opportunists are always driven to histories at which voters believe the incumbent is overwhelmingly likely to be an opportunist who will shirk---so much so that the low outside option is preferable. Conversely, every good type has a chance of reaching and then maintaining a high enough reputation to avoid replacement. 

The dichotomy in the two previous paragraphs is summarized in our main result, \autoref{thm: equivalence theorem}: there is no equilibrium that has ``full effort'' (incumbents exert effort in every period) if and only if all equilibria have ``eventual full effort'' (after some point, incumbents exert effort in all subsequent periods). The result crystallizes a sharp tradeoff between the \emph{possibility} of strong discipline (i.e., even in the short run) and the \emph{guarantee} of either discipline or selection in the long run.

\paragraph{Beyond Electoral Accountability.} 
\hypertarget{para:bureaucracy}{Although} we cast our analysis in terms of politicians and voters, our framework applies more broadly to other principal-agent settings with replacement. For example, consider the relationship between politicians and career bureaucrats. Earlier work has modeled bureaucrats as long-lived agents and politicians as short-lived principals subject to electoral turnover \citep{SpillerUrbiztondo94, AT07, HuberTing21}. 
Indeed, the horizon asymmetry in this relationship can be viewed as a consequence of politicians' electoral accountability: the pressure to satisfy voters transmits a short horizon to politicians in their bureaucratic relationships. More broadly, `short-lived' and `long-lived' are relative, defined with respect to the counterparty in a given accountability relationship. 

For this application, our results connect to two themes. First, they speak to the perils of bureaucratic politicization. In environments conducive to strong incentives, there can be bad equilibria characterized by persistent turnover and shirking. Such dynamics offer a cautionary perspective on 
anti-entrenchment efforts: bureaucratic replacement can become a self-fulfilling norm, destroying the reputational incentives that drive performance and resulting in a permanent cycle of mediocrity. This interpretation of our results is consistent with \citepos{AkhtariMoreiraTrucco22} empirical evidence linking bureaucratic replacement and poor performance to politician turnover. By contrast, when incentives are limited, we show that selection is powerful in the long run. This aligns with 
\citepos{BesleyGhatak05} and \citepos{Prendergast07} emphasis on the identification and retention of intrinsically motivated bureaucrats.

\paragraph{Related Literature.} Our paper primarily contributes to and sits at the intersection of two literatures: electoral accountability and reputation formation. 

In the accountability literature, the idea that replacement is an instrument for selection, rather than just controlling moral hazard, was popularized by \citet{Fearon99}.\footnote{We will not attempt to discuss the large literature on accountability with just moral hazard; see \citet{DM17survey} for a survey.} Much of the literature studying selection, surveyed by \citet{Ashworth12} and \citet{DM17survey}, uses models with either just two periods or a two-term limit; exceptions include \citet{BS93}, \citet{Duggan00}, \citet{BDH04}, and \citet{Schwabe09}. 
None of these papers, nor virtually any others, have results on selection leading to voter-optimal outcomes in the long run---in any equilibrium, much less all equilibria, as we find under some conditions. The only exception  we are aware of is {\citet{anesi2022making}.\footnote{In their extension that incorporates adverse selection, \citet{ALR25} study a model with an opportunistic type and a bad type who always shirks. They show that a voter-optimal equilibrium weeds bad politicians out from office; yet, because opportunistic politicians still face moral hazard, equilibria need not yield sustained effort in the long run. We discuss adding bad types to our model in \autoref{sec:discussion}.
Another paper concerned with long-run outcomes is \citet{DugganForand25}; 
they assume that an incumbent's type is publicly revealed after he assumes office, similar to one section of \citet{KVW19} that does not have term limits but assumes types are revealed after two periods in office. Such exogenous type revelation limits (in the latter paper) or even eliminates (in the former) the reputation-building horizon.  \citet{KVW19} focus on how their short-horizon reputation-building affects incumbency advantage. \citet{DugganForand25} focus on how incumbents can use current policies to manipulate future states and policies, and how patient voters may nevertheless attain good outcomes.} They construct one equilibrium of their model that is approximately first best when each politician's type follows an irreducible Markov chain and all agents---politicians and the voter---are almost fully patient. By contrast, we have fully-persistent politician types with fixed discounting, and short-lived (or impatient) voters. 
\citepos{anesi2022making} construction relies on patient voters using sophisticated punishment schemes that are themselves supported by ``grim-trigger'' strategies; in our setting, it is instead the very inability to discipline opportunists that guarantees eventual positive selection.

When our environment is conducive to incentive provision, our construction of an equilibrium in which incumbents shirk even in the long run uncovers a new mechanism for inefficiency. It owes to the dynamics of opportunistic incumbents' effort over their careers, and the resulting replacement despite a favorable reputation described earlier.\footnote{One can contrast this equilibrium feature with the literature on managerial turnover with dynamic contracting; see, for example, \citet{GP12} and references therein. Papers in that literature characterize ex-ante optimal retention and compensation policies for a long-lived principal with commitment power. Optimal contracts often exhibit ``entrenchment'' in the sense that retention standards become more lenient over time to reduce information rents. Such entrenchment also features in certain analyses of accountability, like \citet*{ALR25}.} In \autoref{subsec:bad-longrun} we explain why the mechanism is distinct from that of the inefficient equilibrium of \citet{Myerson06}, which owes to large replacement costs. The coexistence of good and bad equilibria in our incentive-conducive environments may also be reminiscent of folk theorems in repeated complete-information games \citep[e.g.,][]{fudenberg1986folk}. However, we have short-lived voters who may learn about politician types (so there is incomplete information), which makes the reputation literature discussed next---with quite different prior results---the more relevant comparison.

Most models in the reputation literature have a \emph{fixed} long-lived player interacting with a sequence of short-lived players.
\citet{FL89,FL92} derive tight bounds for the long-lived player's equilibrium payoffs as he becomes fully patient. \citet*{CMS04} show that when monitoring is imperfect, although the \citet{FL92} bounds apply ex ante, the long-lived player will eventually lose his reputation in the sense that his opponents will learn his type.\footnote{\label{fn:CMS-breaking}A number of subsequent papers propose models in which reputation can be sustained in the long run. The mechanisms include competition between multiple long-lived players \citep[e.g.,][]{Horner02,DebFanning25}, the long-lived player's type changing over time \citep*[e.g,][]{phelan2006public, wiseman08, EGW12}, and the short-lived players having limited memories \citep[e.g.,][]{liu2014limited,Pei2024}.} By contrast, our model has multiple long-lived players, and our focus is on their endogenous replacement. The possibility of replacement implies that the short-lived players may not learn  any given long-lived player's type, and because the replacement probabilities are endogenous, the long-lived players' effective discount factors are endogenous and evolving, possibly bounded away from full patience. We instead focus on outcomes from the short-lived players' vantage. 

Our replacement structure is reminiscent of \citet{fudenberg1987incomplete}, who study a Coasian bargaining model. Similar to our paper, their seller has an ``outside option'', whose value is determined endogenously, to restart the interaction afresh with a new buyer. Their analysis concerns trading dynamics that end with a sale; inter alia, they show how the switching option allows the seller to effectively commit to a price in some equilibria. By contrast, we study an infinitely repeated reputation model with moral hazard. We analyze how replacement motivates effort and sorts types, with a focus on characterizing when these forces assure good outcomes in the long run.

There is also prior work on reputation models with competition between long-lived players, notably \citet{Horner02} and, more recently, \citet{DebFanning25}. Besides other differences (e.g., in the behavior of committed types and, in the former paper, endogenous pricing), an important distinction is the nature of their short-lived players' outside option. 
In these competitive models, the consumers' outside option is the equilibrium value of switching to a rival firm (whether a previously established firm or a new entrant), which evolves with the history and produces alternative forces to simply restarting afresh.\footnote{Naturally, considerations of when to restart afresh are not present when the long-lived player exits according to some exogenous process \citep*[e.g.,][]{mailath2001wants,Tadelis02}.}

Finally, ongoing work by \citet{Ali25} studies a model related to ours. He has only two politicians (among whom voters may rotate), and studies a class of equilibria with a ``pure Markovian'' structure. 
While \citet{Ali25} also finds conditions under which all his equilibria have eventual full effort, the mechanism does not operate through selection; to the contrary, a full-effort equilibrium exists under his conditions. \autoref{sec:discussion} elaborates.

\section{Model}\label{sec: model}

\paragraph{Setting.} We consider an infinite-horizon game in discrete time, with periods indexed by $t\in \Zpos\equiv \{0,1,2,\ldots\}$.
There is a pool of ex-ante identical long-lived players,  whom we refer to as \textit{politicians}. Each period, only one long-lived player, whom we call the \textit{incumbent}, is active, or synonymously, holds office. An incumbent remains active until he is {replaced}, after which he is never active again.
There is also an infinite sequence of short-lived players, whom we call \textit{voters}. The voters arrive one in each period, and each voter plays the game only in the period they arrive.

In each period other than the initial period, the current voter decides whether to \emph{keep} the previous period's incumbent or \emph{replace} him with a new politician. In the initial period $0$ (there being no prior incumbent), the voter makes no choice.} (We will typically suppress the qualifier ``current'' hereafter.) 
Then, the incumbent chooses an action $a \in \Act:= \{\leff,\heff\}$. We refer to $a=0$ as \emph{shirking} and $a=1$ as, interchangeably, \emph{working} or \emph{exerting effort}. This action stochastically generates a public signal $s$ drawn from a finite set $S$ according to a monitoring structure $\mon:A\to\Delta S$. Let $\mon_a(s)$ denote the probability of signal $s$ conditional on action $a\in \Act$, and let $\mon_{\tilde a}(s):=\tilde a\mon_1(s)+(1-\tilde a)\mon_0(s)$ denote the aggregate probability of signal $s$ if the incumbent exerts effort with probability $\tilde a\in[0,1]$.

When a politician is active, his stage-game payoff from action $a$ is $\upo(a)$. 
Once a politician is replaced, he is never again active and receives a fixed payoff, which we normalize to $0$. 
The voter's payoff from action $a$ is $\uvo(a)$, and they further incur an additive cost $\cost \geq 0$ if they replace the incumbent.\footnote{We allow for a positive replacement cost to clarify that our analysis is not driven by excessive voter indifference, e.g., being indifferent to replacement when they expect the same incumbent behavior after retention or replacement. We will be interested in the case where the replacement cost is not too large.} We assume that the voter's preferred action is different from the incumbent politician's (myopically) favorite action, and that a politician's least favorite outcome is to be replaced. Without further loss of generality, we normalize payoffs so that $\uvo(a)=a$ and $\upo(a)=1-\ecost a$ for some $\ecost\in(0,1)$. We also assume the monitoring structure is informative, that is, $\mon_0(\cdot)\neq\mon_1(\cdot)$. Finally, assume no signal perfectly reveals the politician to be shirking, that is, $\mon_\heff(s)>0$ for every $s\in S.$

Following \citet{Besley06} and \citet{BesleySmart07}, among others, we take seriously that some politicians are intrinsically motivated. With probability $\prior \in (0,1)$, a given politician is a \textit{good type} who always plays $a=\heff$ when active. With complementary probability $1-\prior$, he is an \textit{opportunistic type} who maximizes the expectation of his discounted average payoff $(1-\delta) \sum_{t=\tau_0}^{\tau_1-1} \delta^{t-\tau_0} \upo(a_t)$, where $a_t$ is his period-$t$ action, $\tau_0$ is the period that he becomes active, $\tau_1$ is the period that he is replaced (which we label as $\infty$ if he is never replaced), and $\delta \in (0,1)$ is his discount factor. Each politician's type is his private information, and types are independent. Hence, after a voter replaces an incumbent, a newly drawn politician is the good type with probability $\prior$. We assume that the replacement cost (which can be zero) is small relative to this uncertainty: $\cost<\min\{\prior,1-\prior\}.$

\paragraph{Interpretations.} 
An incumbent's choice of effort can be interpreted broadly to capture various familiar conflicts of interest in political agency. Working $(a=1)$ can correspond to costly activities that benefit the electorate, such as the effective provision of public goods. Conversely, shirking $(a=0)$ can represent rent-seeking behavior, such as engaging in corruption or implementing inefficient pork-barrel projects that benefit special interests.

We model voters as short lived for two reasons. First, it is a way to capture an electorate composed of many long-lived but uncoordinated voters. Rather than a mass electorate coordinating on complex strategies guided by forward-looking considerations, we take the view that each voter simply votes for the alternative that maximizes their payoff in the current period. Such myopic behavior---which can also be viewed as an equilibrium refinement with many voters, none of whom is pivotal---is mathematically equivalent to a sequence of short-lived voters. Second, we are ruling out sophisticated forward-looking coordination schemes between politicians and voters. Ultimately, as in the canonical reputation literature, our assumption sharpens the focus on how opportunistic politicians' career concerns, rather than voters' long-term considerations, shapes outcomes. The assumption is also appropriate for other settings in which principals are transient, such as the relationship between elected officials and the bureaucracy discussed in the \hyperlink{para:bureaucracy}{Introduction}. We return to this issue in \autoref{sec:discussion}.

\paragraph{Strategies and equilibria.} We will focus on equilibria in which players' behavior depends only on the current officeholder's past outcomes. To that end, the \emph{(public) career history} of a given incumbent records the sequence of public signals he has generated in the past. Let $\Hist:=S^{<\infty}=\bigsqcup_{t=0}^\infty S^t$ denote the set of career histories.\footnote{The full public history---the full history available to voters---consists of an incumbent career history, all past replacement decisions, and the career histories associated with all previously replaced politicians.} Note that $\Hist\setminus\{\emptyset\}=\Hist\times S$. Our solution concept entails a politician strategy $\pstrat:\Hist\rightarrow[0,1]$, a voter strategy $\vstrat:\Hist\setminus\{\emptyset\}\rightarrow[0,1]$, and a voter belief map $\belief:\Hist \rightarrow[0,1]$. Here, $\pstrat(h)$ denotes the probability that an opportunistic incumbent works when his career history is $h$; $\vstrat(h)$ denotes the probability the current voter replaces an officeholder whose career history is $h\neq \emptyset$; and $\belief(h)$ is the probability voters assign to a politician being the good type (his \emph{reputation}) when his career history is $h$. Given such a triple, it is useful to record the voter-expected incumbent efforts and the opportunistic incumbents' continuation values. Formally, suppressing their dependence on $(\pstrat,\vstrat,\belief)$ for brevity, we define $\eeff:\mathcal H\to [0,1]$ via
$$\eeff(h):=\belief(h)+\brac{1-\belief(h)}\pstrat(h),$$
and $\cval:\mathcal H\to [0,1]$ as the unique solution to the recursive equation 
$$\cval(h)=(1-\delta)\brac{1-\ecost\pstrat(h)}+\delta\sum_{s\in S}
\mon_{\pstrat(h)}(s) \brac{1-\vstrat(h,s)} \cval(h,s),\ \forall h\in\Hist.
$$

We call a triple $(\pstrat,\vstrat,\belief)$ a \hypertarget{def:eqm}{\textbf{personal symmetric weak perfect Bayesian equilibrium}} (hereafter \textbf{equilibrium}) if it satisfies the following three properties: \begin{enumerate}
    \item Any $h\in\Hist$ has $$\pstrat(h)\in\arg\max_{a\in[0,1]}\curlyb{(1-\delta)(1-\ecost a) + \delta\sum_{s\in S}\mon_a(s)\brac{1-\vstrat(h,s)}\cval(h,s)}.$$
    \item Any $h\in\Hist\setminus\{\emptyset\}$ has $\vstrat(h)\in\arg\max_{x\in[0,1]}\curlyb{(1-x)\eeff(h) + x\brac{\eeff(\emptyset)-\cost}}$.
    \item \label{eqm3} We have $\belief(\emptyset)=\prior$, and any $h\in\Hist$ and $s\in S$ with either $h=\emptyset$ or $\vstrat(h)<1$ has $${\mon_{\eeff(h)}(s)}\belief(h,s)=
    {\belief(h)\mon_{\heff}(s)}.
    $$
\end{enumerate}

The first equilibrium condition is a recursive formulation of sequential rationality for incumbents.  
The second condition is for voters: being short-lived, they make replacement decisions that maximize 
the current effort they expect from an officeholder net 
of replacement cost. The third condition incorporates voters' Bayesian updating: it requires that (i) any newly-installed incumbent has reputation $\pi_0$, the prior probability of a good type; and (ii) reputation at any career history $(h,s)$ is derived from Bayes' rule applied to $h$ so long as either $h=\emptyset$ (so $s$ is the first signal generated by the incumbent) or the incumbent has been retained with positive probability at $h$.\footnote{We refer to our solution concept as ``weak'' PBE because beliefs can be arbitrary at certain non-initial career histories that are off path---specifically, at any $(h,s)$ with $h\neq \emptyset$ and $\sigma_V(h)=1$. It is still stronger than textbook definitions of weak PBE that impose no belief restrictions off path \citep[Definition 9.C.3, p.~285]{MWG95}. \autoref{subsec:main-theorem} elaborates on our concept vs.~``full'' PBE.}

Following \citepos{BS93} terminology, the ``personal'' nomenclature in our equilibrium concept reflects that the interaction between any officeholder and the voters depends only on the sequence of signals generated by that politician---not, for example, on the calendar time at which the politician took office nor on what transpired prior to when he took office.
\footnote{We have also imposed that the politician conditions only on his own \emph{public} career history, not on his own past actions. This restriction is for notational simplicity: Because the game has ``product monitoring'', standard arguments show this feature is without loss in terms of outcome equivalence.} The ``symmetric'' adjective reflects that all opportunistic officeholders respond to their career history in the same way and are treated identically by voters.  
Combined, these two properties reflect the premise that the role of replacement is to start the interaction with a new politician afresh; put differently, any equilibrium has a single (endogenously-determined) ``outside option''---a voter's value from replacement.
This approach is standard in the political-economy literature; besides \citet{BS93}, see, for example, \citet{Duggan00} or \citet{VanWeelden13}. It is also tantamount to \citepos*{fudenberg1987incomplete} ``stationarity'' restriction. As we shall see, it does not preclude rich dynamics within any 
incumbent's career.

We highlight that our setting is \emph{not} a ``trust game'' in which a principal's refusal to trust effectively ends the stage-game interaction. That game would trivially have a ``no trust, no effort'' equilibrium. By contrast, replacement in our model installs a new incumbent who makes an active effort choice. There is no ``always replace, never work'' equilibrium: the presence of good types implies that if opportunists were never to work, there would be learning and a voter would never replace an incumbent at a history with favorable reputation (i.e., at $h$ such that $\belief(h)>\prior$).

As a preliminary step, we ensure---non-constructively---equilibrium existence despite our refinements.

\begin{Proposition}\label{prop: eqm exists}
An equilibrium exists.
\end{Proposition}

The proof in \appendixref{sec: proof that eqm exists} constructs an auxiliary game whose perfect Bayesian equilibria (PBE) are equilibria of our game.
The auxiliary game has only one politician. It mirrors the interaction in our game between a given politician and the voters, modifying it in two ways. First, the politician's first-period choice is $\aprob\in [0,1]$, representing the probability with which he chooses action $\heff$; in all other periods, the politician makes a choice from $A\equiv \{0,1\}$. While the politician can mix in subsequent periods,  we require him to not mix (over the action set $[0,1]$) at the initial history.
Second, the voter's payoff when they replace the politician is $\prior+\aprob(1-\prior)-\cost$.  Intuitively, our specification of the voter's payoff upon firing the politician captures the symmetry requirement in the original game. We require the politician not to mix in the first period to avoid voters updating beliefs as play progresses about their payoff from replacement (contrary to symmetry in the original game). Making the politician's first-period choice continuous secures the existence of such PBE in the auxiliary game, following standard arguments \`a la \citet{kreps1982sequential} and \citet{fudenberg1983subgame}.

\section{Results}
\label{sec:results}

Our goal is to study whether replacement ensures that incumbents exert effort. To that end, we introduce two criteria.  Using bold font for random variables, let $\actrv_t\in\{0,1\}$ indicate the incumbent's chosen effort in period $t$, which will in particular be equal to $1$ if the incumbent is the good type. We say that an equilibrium \textbf{attains full effort} if $\actrv_t=1$ almost surely for every $t\in\Zpos$, and it \textbf{attains  eventual full effort} if there almost surely exists some $\boldsymbol{\tau}\in\Zpos$ such that $\actrv_t=1$ for all $t\geq \boldsymbol{\tau}$. In words, an equilibrium attains full effort (FE) if incumbents always work (on the path of play), and it attains eventual full effort (EFE) if, eventually, incumbents work in every period.\footnote{One might be interested in weaker variants of both FE and EFE. Not only would our positive results obviously extend, but, as elaborated in Subsections \ref{subsec:FEIandFE} and \ref{subsec:bad-longrun}, so do our negative results. 
} The distinction between FE and EFE is thus one of timing; an equilibrium with FE also has EFE, but the converse is not generally true. In addition, one can also draw a distinction between whether, for any given parameters, either property holds in \emph{no} equilibrium, \emph{some} equilibrium, or \emph{all} equilibria.

A strong desideratum is that all equilibria attain FE; conversely, it would be troubling if no equilibrium attains even EFE. We will see that neither of these is possible, no matter the parameters.\footnote{One can show that absent good types, no equilibrium attains EFE when \condref{FEI} below fails. Even with a good type, no equilibrium attains EFE in the canonical noisy-monitoring product-choice game between a long-lived firm and short-lived consumers when there is no possibility of replacement \citep*{CMS04}. In a different model, \citet{Pei2024} presents a result in the vein of all equilibria attaining FE.} Instead, our main result is that parameters divide sharply into two cases: either there is an equilibrium that attains FE, but then there is also an equilibrium that fails even EFE; or there is no equilibrium that attains FE, but then all equilibria attain EFE. The following table summarizes.

\begin{table}[h!]
\centering
\begingroup
\setlength{\tabcolsep}{10pt}
\renewcommand{\arraystretch}{1.5}
\begin{tabular}{r @{\hspace{15pt}} @{}c c @{} c@{}}
\toprule
 & \text{Effort always (FE)} & \text{Effort eventually (EFE)} \\
\midrule
\text{Condition FEI} 
  & \textit{Some} equilibrium
  & \textit{Some} equilibrium \\
\text{$\lnot$ Condition FEI} 
  & \textit{No} equilibrium
  & \textit{All} equilibria \\
\bottomrule
\end{tabular}
\endgroup
\caption{The equilibrium possibilities. Full Effort Incentives (FEI) is a condition on parameters.}
\label{table:possibilities}
\end{table}

The table's \condref{FEI} on parameters is the following:

\begin{FEI}\label{FEI}
Some $\cval: S\to \Rpos$ exists such that
\begin{equation}\label{IC for FB}
 (1-\delta)(1-\ecost) +\delta \sum_{s \in S} \mon_1(s)\cval(s) \geq 
 (1-\delta) 1 +\delta\sum_{s \in S} \mon_0(s)\cval(s) \tag{$\text{IC}_{\text{FEI}}$}
\end{equation}
and
\begin{equation}\label{PK for FB}
 (1-\delta) (1-\ecost) +\delta\sum_{s \in S} \mon_1(s)\cval(s) \geq \max_{s \in S} \cval(s).\tag{$\text{PK}_{\text{FEI}}$} 
\end{equation}
\end{FEI}

\condref{FEI} is independent of politicians' initial reputation $\prior$ and voters' cost of replacement $\cost$ because the condition stems from contemplating an auxiliary game without good types and without a replacement cost. 
In the auxiliary game, if all incumbents always exert effort, then voter incentives are trivially satisfied. Standard recursive arguments then imply that under some conditions, there are equilibria in which appropriate voter behavior can incentivize incumbents to always exert effort. The associated incentive-compatibility and promise-keeping constraints are precisely those in \condref{FEI}, inequalities \eqref{IC for FB} and \eqref{PK for FB} respectively. \condref{FEI} is satisfied if and only if the politicians' discount factor $\delta$ is high enough, their effort cost $\ecost$ is small enough (recall we have normalized their benefits of holding office to $1$), and the monitoring structure $\mon$ is sufficiently informative.\footnote{We emphasize that \condref{FEI} is a joint condition on the parameters $(\delta,\ecost,\mon)$. So, for instance, the required discount factor is higher when the effort cost is higher or the monitoring structure is less informative in the \citet{Blackwell53} sense. As an illustration, with binary signals of precision $p\in (1/2,1)$ (that is, $S=\{0,1\}$ and $\mon_a(a)=p$), it can be verified using \autoref{lem: cutoff rule for first-best equilibrium}---which shows that, no matter the monitoring structure, there no loss in restricting to $\cval(s)\in\{0,\bar\cval\}$ for some $\bar\cval>0$---that \condref{FEI} reduces to $$\delta\geq \frac{\kappa}{p-(1-p)(1-\kappa)},$$
where the right-hand side is decreasing in $p$ and increasing in $\ecost$.} 
Parameters that satisfy \condref{FEI} thus correspond to environments that are conducive to incentive provision.

\subsection{Main Result}
\label{subsec:main-theorem}

We now present our main result.

\begin{Theorem}\label{thm: equivalence theorem}
The following are equivalent:

\begin{enumerate}
  \item \label{allEFE} All equilibria attain eventual full effort;
  \item \label{noneFB} No equilibrium attains full effort;
  \item \label{FBIfails} \condref{FEI} fails.
\end{enumerate}
\end{Theorem}

An implication is that no matter the parameters, there is always at least one equilibrium that attains EFE, but also at least one equilibrium without FE.

We view the main economic lesson from the theorem as the sharp contrast summarized in \autoref{table:possibilities}: environments conducive to incentive provision yield the \emph{possibility} of FE; those that are not conducive \emph {guarantee} EFE.  More precisely, when \condref{FEI} holds, although there are equilibria that attain FE, there are also equilibria that fail EFE.  By contrast, when \condref{FEI} fails, there is no FE equilibrium, but perpetual effort is inevitable after finite time in {every} equilibrium.  As we will elaborate, EFE is guaranteed through (eventual) selection of a good type; paradoxically, when \condref{FEI} holds, the very possibility of providing strong incentives undermines assuring good selection or assuring good long-run outcomes. To our knowledge, this is a novel tradeoff between discipline and selection.

The proof of \autoref{thm: equivalence theorem} is in \appendixref{sec: thm proof appendix}. The remainder of this section discusses the intuition and forces underlying the theorem. But before that, we comment on an aspect of using (personal, symmetric) \hyperlink{def:eqm}{\emph{weak} PBE} as our solution concept. Condition \eqref{eqm3} of the equilibrium definition allows for arbitrary beliefs {off path} immediately after a voter retains an incumbent who should have been certainly replaced. Dropping the condition's qualification that $\vstrat(h)<1$ would impose belief requirements off path that yield (personal, symmetric) PBE, without the ``weak'' moniker. \appendixref{sec: proof that eqm exists} assures existence for this solution concept as well. Our equilibrium constructions under \condref{FEI} benefit from the off-path latitude allowed by weak PBE, as flagged in \appendixref{sec: thm proof appendix}. At the same time, the weaker solution concept strengthens what we view as the most important conclusion of \autoref{thm: equivalence theorem}, namely that all equilibria attain EFE when \condref{FEI} fails.

\subsection{\condref{FEI} and Full Effort}
\label{subsec:FEIandFE}

The equivalence between \condref{FEI} and the existence of FE equilibria stems from there being no learning in an FE equilibrium. So there is an FE equilibrium if and only if there is one in an auxiliary game without commitment types. It turns out that for a small enough replacement cost, it is further equivalent to consider the auxiliary game without a replacement cost. As mentioned earlier, \condref{FEI} characterizes when there are FE equilibria in that auxiliary game; the logic follows standard recursive reasoning \citep{APS90}. In fact, when \condref{FEI} fails and replacement costs are sufficiently small, there is a positive lower bound (across all equilibria) on the probability with which opportunistic incumbents shirk in their first period in office; see \autoref{cor: equivalence theorem with extra condition} in the Appendix. This implies that when \condref{FEI} fails, there is some $\epsilon>0$ such that no equilibrium has $\mathbb{P}\left(\actrv_t=1\right)\geq 1-\epsilon \text{ for all } t\in\Zpos$; indeed, the inequality already fails for $t=0$.

\subsection{What Prevents Eventual Full Effort?}
\label{subsec:bad-longrun}

Our argument that there is an equilibrium without EFE when \condref{FEI} holds is constructive. Notice that when replacement and effort costs are small, and monitoring is very precise, a newly-installed opportunistic incumbent has to work with positive probability: if not, a voter would retain an incumbent if he has a favorable reputation (i.e., a posterior $\belief_t$ strictly above the prior $\prior$), which would incentivize newly installed incumbents to work toward such a reputation. At least when the replacement cost is zero, a newly-installed incumbent cannot work with certainty either: otherwise, voter incentives imply that any retained incumbent must be working, which implies EFE.

With that in mind, let us describe the non-EFE equilibrium we construct. For simplicity, suppose there are only two signals: a ``Pass'' signal and a ``Fail'' one, with the former having the higher likelihood ratio ${\mon_1(s)}/{\mon_0(s)}$. Opportunistic incumbents work with interior probability in their first period in office. Thereafter, if they have always generated a Fail signal, they are replaced with an interior probability and, if retained, work with an interior probability.\footnote{Since reputation is deceasing in the number of Fail signals, the opportunistic type's probability of work increases to keep voter incentives fixed. Meanwhile, the voter's replacement probability is constant to keep politician incentives fixed. Although the underlying forces are quite different, the idea that opportunistic incumbents misbehave less over time is reminiscent of the literature on the repeated prisoner's dilemma with voluntary separation \citep[e.g.,][]{GR96,Eeckhout06} and the ``starting-small'' literature \citep[e.g.,][]{watson1999starting}.} At any history with only Pass signals, the incumbent is retained and works. At histories where a Pass signal is followed by a Fail signal, the incumbent is replaced. See \autoref{fig: bad equilibrium construction} for an automaton depiction of the strategy profile. We defer to the \appendixlink \ the explanation of why this profile constitutes an equilibrium. Crucially, however, this equilibrium has the features that every incumbent is eventually replaced, and incumbents in their first period shirk with positive probability.  Hence, the equilibrium does not attain EFE; in fact, almost surely there is shirking infinitely often, and average voter welfare (even gross of replacement costs) is below first best:
$$\limsup_{T\to\infty} \E\brac{\frac{1}{T} \sum_{t=0}^{T-1} \actrv_t}<1.$$
 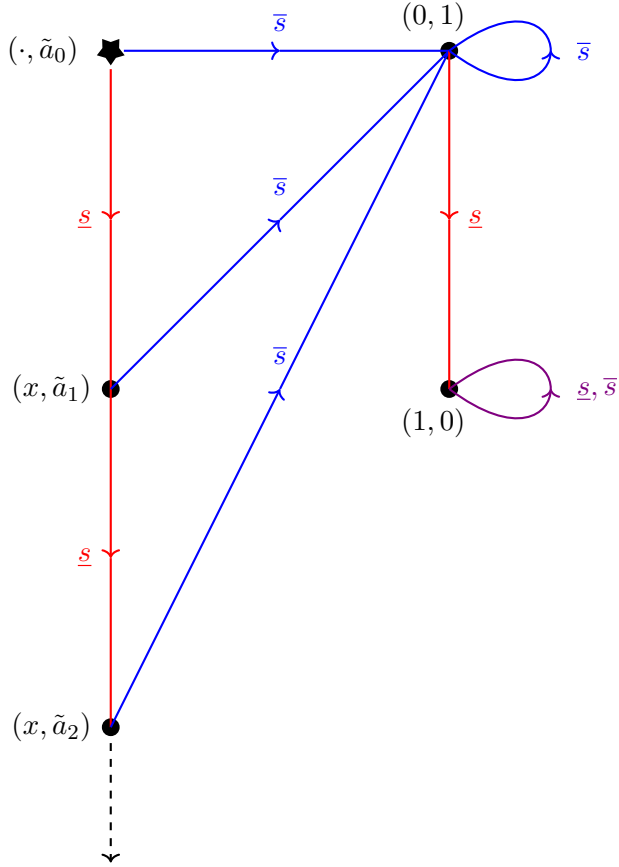
\begin{figure}
\begin{center}
\tikzset{
  blue loop arrow/.style={
    blue,
    postaction={decorate}
  },
  purple loop arrow/.style={
    violet,
    postaction={decorate}
  },
  decoration={markings, mark=at position 0.5 with {\arrow{>}}}
}

\begin{tikzpicture}[scale=0.9, every node/.style={font=\small}, every path/.style={thick}]
\node[star, star points=5, star point ratio=.6, fill=black, ultra thick] (start) at (0,0) {};
\coordinate (n11) at (5,0);
\coordinate (n01) at (0,-5);
\coordinate (n02) at (0,-10);
\coordinate (n00) at (5,-5);

\filldraw[black, ultra thick] (n11) circle [radius=0.1];
\filldraw[black, ultra thick] (n01) circle [radius=0.1];
\filldraw[black, ultra thick] (n02) circle [radius=0.1];
\filldraw[black, ultra thick] (n00) circle [radius=0.1];

\node[left=3pt of start] {$(\cdot,\aprob_0)$};
\node[above, xshift=-6pt, yshift=3pt] at (n11) {$(0,1)$};
\node[left=3pt] at (n01) {$(x,\aprob_1)$};
\node[left=3pt] at (n02) {$(x,\aprob_2)$};
\node[below, xshift=-6pt, yshift=-3pt] at (n00) {$(1,0)$};

\draw[->, blue] (start)--(2.5,0) node[above, yshift=3pt]{$\overline{s}$};
\draw[blue] (2.5,0)--(n11);
\draw[->, red] (start)--(0,-2.5) node[left=3pt]{$\underline{s}$};
\draw[red] (0,-2.5)--(n01);
\draw[->, red] (n01)--(0,-7.5) node[left=3pt]{$\underline{s}$};
\draw[red] (0,-7.5)--(n02);
\draw[dashed, ->] (n02)--(0,-12);
\draw[->, blue] (n01)--(2.5,-2.5) node[above=5pt]{$\overline{s}$};
\draw[blue] (2.5,-2.5)--(n11);
\draw[->, blue] (n02)--(2.5,-5) node[above=5pt]{$\overline{s}$};
\draw[blue] (2.5,-5)--(n11);
\draw[->, red] (n11)--(5,-2.5) node[right=3pt]{$\underline{s}$};
\draw[red] (5,-2.5)--(n00);

\draw[blue loop arrow] (n11) .. controls +(2,-1.5) and +(2,1.5) .. (n11) node[right=44pt]{$\overline{s}$};
\draw[purple loop arrow] (n00) .. controls +(2,-1.5) and +(2,1.5) .. (n00) node[right=44pt]{$\underline{s},\overline{s}$}; 
\end{tikzpicture}

\caption{Under \condref{FEI}, an equilibrium that does not attain eventual full effort when, for simplicity, there are only two signals, $\underline{s}$ and $\overline{s}$, with $\mon_1(\underline s)/\mon_0(\underline s)<\mon_1(\overline s)/\mon_0(\overline s)$. Each node (black star or dot) represents an automaton state, with the initial state---corresponding to a new incumbent---indicated by the star.  The vector at each node shows the corresponding strategy profile: the first entry is the probability that the incumbent is replaced (not defined at the initial state), and the second entry is the probability that an opportunistic incumbent exerts effort if retained.
Each edge represents the state transition following the corresponding signal. The equilibrium satisfies $0<\aprob_0<\aprob_1<\cdots<1$ and $x\in (0,1]$.
}
\label{fig: bad equilibrium construction}
\end{center}
\end{figure}

The mechanism underlying the above equilibrium is a novel rationale for undesirable long-run outcomes in accountability models. In particular, it ought to be contrasted with \citet{Myerson06}. For some parameters of his model, his Theorem~1 exhibits equilibria in which opportunistic incumbents always shirk.\footnote{\citepos{Myerson06} model for unitary democracy is like ours, except that he assumes a single long-lived and forward-looking voter who can perfectly observe the incumbents' effort choices, and he allows for large replacement costs. But the relevant part of his Theorem 1 also holds with short-lived voters who can only observe noisy signals, as in our model.}
Yet incumbents are never replaced (and hence EFE does not obtain) because the cost of replacement is large relative to the probability of a good type. We are interested, instead, in small replacement costs. In our non-EFE equilibrium, voters replace incumbents at some histories despite a favorable reputation (i.e., when the posterior $\belief_t$ strictly exceeds the prior $\prior$) because of the rational expectation that an opportunistic incumbent at such a history will shirk.

In fact, this phenomenon of replacing an incumbent despite a favorable reputation is a necessary feature of any equilibrium that does not attain EFE. Formally, we have:

\begin{Proposition}\label{prop: firing happens with a good reputation}
Any equilibrium that does not attain eventual full effort entails a positive probability of an incumbent being replaced with a reputation greater than the prior. 
\end{Proposition}

An intuition for the argument is as follows. As we will see shortly (\cref{prop: differential replacement}, part \ref{prop: opportunistic replaced}), any opportunistic incumbent is eventually replaced. Hence, if EFE fails, then any good incumbent must also eventually be replaced---otherwise, eventually some good incumbent will stay in office forever, yielding EFE. Moreover, at least absent a replacement cost, there must be some learning in an incumbent's first period.\footnote{Otherwise, a new incumbent must be working, which implies FE when there is no replacement cost.} It follows that incumbents must be replaced at some histories with favorable reputation.

\subsection{What Guarantees Eventual Full Effort?}
\label{subsec:good-longrun}

We now turn to why a failure of \condref{FEI} implies that all equilibria attain EFE. The key is that precisely because the environment is not conducive to incentive provision, selection is now ineluctable. The following result, a key input to \cref{thm: equivalence theorem},  formalizes the point.

\begin{Proposition}\label{prop: differential replacement}
Consider any equilibrium.
\begin{enumerate}
    \item \label{prop: opportunistic replaced} Every opportunistic officeholder is eventually replaced.
    \item If \condref{FEI} fails, then every good officeholder has a positive probability of being retained forever. 
\end{enumerate}
    Hence, if \condref{FEI} fails, then some good officeholder is retained forever.
\end{Proposition}

The last statement of \autoref{prop: differential replacement} follows from the first two parts and an application of the Borel-Cantelli lemma.

The logic for the first part of the proposition is reminiscent of \citet*{CMS04}. Here is the  intuition. Since any incumbent's reputation is a martingale, it converges almost surely. If an opportunistic incumbent is not eventually replaced, then because signals statistically identify effort, eventually either (i) he must be exerting effort---behaving like the good type, so that belief updating stops---or (ii) his type must be revealed. Case (i) contradicts him not being eventually replaced, because it is only the possibility of replacement that provides effort incentives. The same logic applies if we are in case (ii) and yet the incumbent eventually still exerts effort with positive probability. Finally, if we are in case (ii) and the incumbent is certain to shirk, then it is optimal for voters to replace him.

Now consider the second part of \autoref{prop: differential replacement}. When \condref{FEI} fails, full-effort incentives cannot be sustained; in particular, as noted earlier, a newly installed incumbent shirks with positive probability. Thus, once an incumbent's reputation is sufficiently high, a voter strictly prefers retention to replacement, since replacement exposes her to some chance of shirking. Put differently, there is a reputation threshold above which the voter never replaces. The key step is that these high no-replacement beliefs are reached with positive probability. In fact, we show that reputation can get arbitrarily high. The logic is that if the supremum of on-path beliefs were interior, then at beliefs close to that supremum an opportunistic incumbent would have to exert effort with high probability---otherwise some signal would push the posterior above the supposed supremum, contradicting that being a supremum. But such effort is intuitively at odds with the failure of \condref{FEI},\footnote{This is only an intuition because \condref{FEI} does not actually preclude the possibility of effort with certainty at some histories. The formal argument is substantially more involved.} and so the supremum must equal one. Now, once beliefs enter the high region where there is no replacement, martingale arguments imply that there is a positive probability of never leaving that region. Combined with the first part of the proposition, it follows that every good-type officeholder is retained forever with positive probability when \condref{FEI} fails.

\autoref{prop: differential replacement} implies that when \condref{FEI} fails, every equilibrium asymptotically yields voters their first best welfare, even net of replacement cost. More precisely:

\begin{Corollary}\label{cor: long run voter welfare}
Suppose \condref{FEI} fails.  
In any equilibrium, the period-$t$ voter's expected payoff converges to first best as $t\to\infty$.
\end{Corollary}

We note that the selection logic of \autoref{prop: differential replacement} holds even if the good type were to exert effort with probability less than $1$. 
The welfare conclusion of \autoref{cor: long run voter welfare} would, of course, be attenuated as that probability decreases.

Although \autoref{prop: differential replacement} shows that selection is extremely powerful in the long run in every equilibrium (when \condref{FEI} fails), it prompts a question about the short run. To assure that opportunists are eventually replaced, what makes replacement attractive to voters in the first place---especially when good types may be rare? One might guess that the voters' endogenous \textbf{outside option}---their equilibrium payoff from replacing an incumbent, gross of replacement cost---must be high for selection to operate effectively.\footnote{More precisely, we define the outside option as the probability of effort from a newly-installed incumbent, which can be written as $\E[\actrv_0]=\eeff(\emptyset)=\prior+(1-\prior)\pstrat(\emptyset)$.} In fact, the following result says that the exact opposite is true when good types are scarce and replacement costs are small.

\begin{Proposition}\label{prop: outside option can be super low}
Suppose \condref{FEI} fails. For any $\varepsilon>0$, there exist $\bar\cost>0$ and $\bprior\in(0,1)$ such that, if $\cost\leq\bar\cost$ and $\prior\leq\bprior$, then every equilibrium has voters' outside option below $\varepsilon$.\footnote{Recall that \condref{FEI} is a condition only on $(\ecost,\delta,\mon)$, not on $(\cost,\prior)$. We show that for any $(\ecost,\delta,\mon)$ violating \condref{FEI}, an appropriate pair $(\bar\cost,\bprior)$ can be found to deliver the conclusion.}
\end{Proposition}

In other words, assuming \condref{FEI} fails, the combination of scarce good types and low replacement costs implies that voters' outside option is arbitrarily low---or equivalently, newly-installed opportunists are almost certainly shirking. So, replacement can look arbitrarily unattractive in the short run, despite culminating powerfully in the long run. While we defer the logic of \autoref{prop: outside option can be super low} to the \appendixlink, why does a low outside option not undermine eventual selection? It is because opportunists are assured to reach histories at which voters believe the incumbent is overwhelmingly likely to be an opportunist who will shirk---so much so that even the low outside option is attractive enough to trigger replacement, given small replacement costs.

\section{Discussion}
\label{sec:discussion}

This section discusses some extensions of our model.

\paragraph{Bad types.} The possibility that some politicians are good types---intrinsically motivated to exert effort in office---is essential to our results. If all politicians are opportunistic, then although \condref{FEI} still characterizes when full effort can be attained, there is always an equilibrium (regardless of parameters) in which incumbents always shirk. 

What if we allow for the coexistence of good types with  \emph{bad} types who always shirk in office? In addition to suitably adapting our maintained assumption of replacement costs being small enough, \condref{FEI} must also be strengthened for the existence of an equilibrium in which  opportunistic incumbents always exert effort. The reason is that now some sequences of signals will convince voters that the incumbent is likely a bad type and must be replaced, which makes it more difficult to sustain effort from the opportunistic type along such sequences. 
More interestingly however,
when \condref{FEI} fails it remains true that all equilibria attain eventual full effort. The mechanism still operates via selection of the good type, with an argument analogous to that of \autoref{prop: differential replacement}. In particular, following the logic of \citet*{CMS04}, not only will every opportunistic officeholder eventually be replaced, but so too every bad type.

\paragraph{Exogenous turnover.} We have assumed that politicians remain eligible for office indefinitely. Plainly, this is essential for \autoref{prop: differential replacement}'s conclusion that when \condref{FEI} fails, some good type will eventually stay in office forever. 

Suppose now that incumbents exit office exogenously with probability $1-\survprob$ in each period (e.g., due to retirement or death), for some survival probability $\survprob\in (0,1)$. Beyond implying that every incumbent leaves office at some time almost surely, politicians' discount factors are also effectively reduced to $\delta \survprob$. Naturally, then, the relevant version of \condref{FEI} now replaces $\delta$ with $\delta \rho$ in \eqref{IC for FB} and \eqref{PK for FB}. Our full-effort equilibrium construction remains valid when this version of \condref{FEI} holds. However, when this version of \condref{FEI} fails, eventual full effort cannot be attained in any equilibrium---the eventual exit of every incumbent creates a friction even after adjusting for the lower effective discount factor. Nevertheless, our conclusion regarding the power of selection is robust in the following sense. Holding the effective discount factor $\delta \rho$ fixed---so, in particular, not altering the aforementioned version of \condref{FEI}---we can ask what happens as $\rho \to 1$. Because voters are short-lived and politician incentives depend only on the effective discount factor, the set of equilibrium strategies does not change. Thus, the equilibrium distributions of the ``search time'' for a good type to take office and never be replaced (barring exogenous exit) is unchanging. As the survival probability $\rho \to 1$, in every equilibrium the system spends a fraction of time approaching $1$ in the hands of a good type, and the long-run average effort probability converges to $1$.

\paragraph{Recalling incumbents.} 
We have assumed that an incumbent who is replaced can never hold office again. This ``no-recall'' assumption is commonplace in the literature, dating back to \citet{BS93}; it was also discussed by \citet{Ferejohn86} in his pure moral-hazard setting. The assumption substantially simplifies tracking the voters' outside option by making it history independent (given the focus on \hyperlink{def:eqm}{personal equilibria}). Even if recall is allowed, \condref{FEI} still characterizes when there is a full-effort equilibrium, and when the condition holds, our construction of an equilibrium without eventual full effort goes through as well. While one might conjecture that allowing for recall would also not change our result that all equilibria attain eventual full effort when \condref{FEI} fails, we do not have a proof.

A related point is that we assume an infinite pool of politicians. With a finite pool, one would have to consider recall. Moreover, for selection to be assured in the long run, one would have to assume the finite pool contains at least one good type---so politician types could not be independent. Such an analysis is beyond the scope of the current paper.\footnote{In ongoing work, \citet{Ali25} studies a model similar to ours, except that he has only two politicians. He identifies conditions under which every equilibrium within a class attains eventual full effort, and shows some of those equilibria attain full effort. By contrast, in our setting \autoref{thm: equivalence theorem} says these two desiderata hold only under disjoint conditions. \citepos{Ali25} eventual full effort logic is not one of selection, unlike ours. One reason for the difference is \citepos{Ali25} focus on pure Markov strategies for incumbents. Our equilibrium that fails eventual full effort entails incumbents mixing; indeed, \autoref{prop: firing happens with a good reputation} implies that every equilibrium of our model in which politicians use pure strategies yields eventual full effort. We see no compelling reason to exclude mixed strategies from consideration; inter alia, as in many models of reputation, one cannot otherwise assure equilibrium existence for all parameters---specifically,  when our \condref{FEI} fails.}

\paragraph{Long-lived voter.} As discussed in \autoref{sec: model}, there are multiple reasons for our assumption of short-lived voters. But now consider instead only one long-lived voter who maximizes her discounted average payoff with a discount factor $\vdiscount \in [0,1)$, which can be different from the politicians' discount factor $\delta$.

When $(\delta,\ecost,\mon)$ satisfies \condref{FEI} and the replacement cost $\cost$ is small, there is still an equilibrium that attains full effort as well as one that fails eventual full effort. 
The former requires no change from our full-effort equilibrium construction.  However, in our equilibrium that fails eventual full effort (\autoref{fig: bad equilibrium construction}), a politician's strategy now needs to be modified to account for the voter's discount factor $\vdiscount$.  Specifically, at non-initial career histories consisting only of Fail signals, the probabilities with which the opportunistic type works in order to keep the voter indifferent depend on $\vdiscount$; but suitable probabilities exist.

When \condref{FEI} is violated, there is still no equilibrium with full effort. However, it is unclear whether eventual full effort attains in all, or even some, equilibria. In particular, it is unclear whether equilibria have a positive probability of retaining each good-type officeholder in the long run, because we cannot rule out that equilibrium reputation is bounded away from $1$.\footnote{In more detail: with a long-lived voter, it is no longer the case that the incumbent needs to be retained with positive probability following some signal $s$ with $\mon_1(s)>\mon_0(s)$. As an example, consider a career history $h$ at which the incumbent shirks: $\pstrat(h)=0$. Unlike when the voter is short-lived, the incumbent being retained with positive probability at $h$ no longer implies that he will be retained with probability $1$ at any career history $h'$ with higher reputation ($\belief(h')>\belief(h)$), as a long-lived voter's discounted average payoff from keeping this incumbent can be lower at $h'$ even when the voter's stage-payoff from doing so is higher at $h'$.}
This echoes the difficulty of obtaining tight predictions on equilibrium payoffs in reputation models without replacement but with two long-lived players \citep{cripps1997reputation,chan2000non}.

\section{Conclusion}
\label{sec:conclusion}

Our motivating question is whether the power of replacement ensures that officeholders act in voters' interests in the long run. Inspired by previous models of accountability and reputation, we have studied the question in a simple and stylized model, one that we believe is a natural benchmark. Our analysis yields a sharp, dichotomous answer: when there are some intrinsically-motivated politicians, replacement makes good long-run outcomes always possible, but only guaranteed in environments that are \emph{not} conducive to incentive provision. 

Our main result, \autoref{thm: equivalence theorem}, provides a new perspective on the tension between sanctions and selection that has been noted by others \citep[e.g.,][]{BS93,Fearon99}. In environments conducive to incentive provision, there is an equilibrium with perpetual effort from incumbents. But politics can alternatively fall into a trap of excessive turnover, driven by voters' rational replacement of incumbents whose favorable past performance now makes them shirk, which bounds 
long-run outcomes away from full effort. In these environments, it is not the underlying conditions that determine long-run outcomes; instead political norms---self-fulfilling expectations---do, even absent any intertemporal voter coordination.

Conversely, in environments where strong incentives cannot be provided (i.e., there is no equilibrium with full effort), the very weakness of accountability becomes the engine of its long-run success. Selection is now inexorable: opportunistic types are weeded out, and voters are assured to eventually secure a good politician who works in their interest.  To return to the view of J.S. Mill quoted in the \hyperref[sec:intro]{Introduction}, in our model it is precisely when opportunistic officeholders cannot be fully disciplined that the worth of the State is eventually determined by the worth of the individuals composing it.

\newpage
\appendix

\section{Proof Appendix}
\label{sec:appendix}

In our proofs below, we do \emph{not} maintain the assumption made in the main text that $\cost<\min\{\prior,1-\prior\}$. Instead, to clarify which results require what (if any) assumption on the replacement cost $\cost$, we state the requisite assumption when needed. 

\subsection{Equilibrium Existence}\label{sec: proof that eqm exists}
As discussed in \autoref{subsec:main-theorem}, our solution concept can be strengthened to a \textbf{personal symmetric perfect Bayesian equilibrium (PBE)} by dropping the qualification of $\sigma_V(h)<1$ in condition \eqref{eqm3} of our \hyperlink{def:eqm}{equilibrium definition}. The proof below shows that that such a PBE exists in our game.

\begin{proof}[Proof of \autoref{prop: eqm exists}]
Let $\Hist$ be as defined for our original game. For any $\varepsilon\in[0,\tfrac12]$ and $T\in\mathbb N\cup\{\infty\}$, define an auxiliary game $\Gamma(\varepsilon,T)$ with a voter in every period and a single politician---who will be a good type with probability $\prior$ and an opportunistic type with probability $1-\prior$. 

We begin by describing $\Gamma(0,\infty)$. The game form is almost exactly as in our original game before the first time a voter replaces the politician, 
except that at the initial history $h=\emptyset$, the politician chooses a continuous action 
$\aprob_0 \in [0,1]$, which generates a signal from $\mon_{\aprob_0}\in\Delta S$. Say that a politician \emph{semi-pure} strategy is $\pstrat:\Hist\to[0,1]$, with the interpretation that the politician is not mixing at the initial history but could be mixing (over the binary action set $\{0,1\}$) at other histories. Voters' (mixed) strategies are given by $\vstrat:\Hist\times S\to[0,1]$. Furthermore, in this modified game, replacing the politician ends the game. The payoffs of the politician are exactly the same as those of the initial officeholder in our original model. If the politician is not replaced, the payoffs of voters in our auxiliary game are exactly the same as those in our original model. But if a voter replaces the politician, then her payoff is equal to $\prior+\paren{1-\prior}\aprob_0$, a function of the unobserved initial choice of the politician.

A key feature of the game $\Gamma(0,\infty)$ is that the politician's choice $\aprob_0\in[0,1]$ directly enters the payoff of the voter who replaces the politician. That is why we treat the politician's choice at the initial history as a continuous action rather than a randomization over the actions $0$ and $1$.\footnote{If the politician were to mix between actions $0$ and $1$ (or any other actions) at the initial history, subsequent voters' incentives would depend on the signal $s_0$, because it would convey information about the realized initial action, and so the value of replacing the politician.}  
With that caveat, we note that every PBE of the auxiliary game $\Gamma(0,\infty)$ in which the politician uses a semi-pure strategy---a semi-pure PBE, for short---corresponds exactly to a personal symmetric PBE of our original model.\footnote{To be precise, our notion of PBE in the auxiliary game is what \citet{Mailath18} refers to as an ``almost perfect Bayesian equilibrium''---although \citet{Mailath18} restricts attention to finite games, which our auxiliary game is not, the definition can be applied essentially verbatim because we have a finite set of signals.
}

So it remains to see that $\Gamma(0,\infty)$ has a semi-pure PBE. The remainder of the proof now follows familiar lines (\`a la \citet{fudenberg1983subgame} and \citet{kreps1982sequential}), but we summarize such an argument for the sake of completeness because the game $\Gamma(0,\infty)$ is not an instance of the models covered in those papers.

To show existence, define for every time horizon $T\in \Zstrpos$ the game $\Gamma(\varepsilon,T)$ by further constraining players' behavior as follows. Before time $T$, all players (specifically, all voters and opportunistic-type politicians) choose from $[\varepsilon,1-\varepsilon]$ at every history rather than from $[0,1]$; and starting at time $T$, all players choose from $\{\tfrac12\}$.

First, consider the case of $\varepsilon>0$ and $T<\infty$. Let us view the game $\Gamma(\varepsilon,T)$ in agent form (with a politician agent and voter agent at each public history). Observe that each agent chooses from a compact interval, only finitely many agents make a nontrivial choice, payoffs are continuous in the strategy profile, and each agent's payoff is affine in their own action. Therefore, by the Debreu-Fan-Glicksberg theorem, a Nash equilibrium $(\pstrat^{\varepsilon,T},\vstrat^{\varepsilon,T})$ exists for $\Gamma(\varepsilon,T)$. Because every history is reached with strictly positive probability under this profile, a unique belief map $\belief^{\varepsilon,T}$ exists that satisfies the Bayesian property given $\pstrat^{\varepsilon,T}$. Also because every history is reached with strictly positive probability, the Nash equilibrium property implies sequential optimality (in the restricted game).

Now, fixing $\varepsilon>0$, consider the sequence $(\pstrat^{\varepsilon,T},\vstrat^{\varepsilon,T},\belief^{\varepsilon,T})_{T=1}^\infty$ from the compact metrizable space $\brac{\paren{[\varepsilon,1-\varepsilon]\times [\varepsilon,1-\varepsilon]^S\times[0,1]}^\Hist}^3$. By compactness, some subsequence converges to some triple $(\pstrat^{\varepsilon},\vstrat^{\varepsilon},\belief^{\varepsilon}).$

Appealing again to compactness, $(\pstrat^{\varepsilon},\vstrat^{\varepsilon},\belief^{\varepsilon})_{\varepsilon\in\left(0,\tfrac12\right)}$ also has some limit point $(\pstrat,\vstrat,\belief)$ as $\varepsilon\to0$.
Continuity of payoffs in the strategy profile (given that the game is continuous at infinity) and continuity of the Bayesian condition (which is easy to see by clearing the denominator) imply $(\pstrat,\vstrat,\belief)$ is a semi-pure PBE of $\Gamma(0,\infty)$. \end{proof}

\subsection{The Main Theorem}\label{sec: thm proof appendix}

\subsubsection{Proof Overview}\label{sec: proof sketch}

As the main theorem's proof is somewhat involved, we first provide an overview of the arguments. To help parse the logical flow of these arguments, \autoref{fig: lemma flowchart} details which results are inputs to which other results in the paper.

\begin{figure}
\resizebox{\textwidth}{!}{%
\begin{tikzpicture}[
  node distance=1cm and 1.5cm,
  box/.style={rectangle, draw, rounded corners, minimum width=1.5cm, minimum height=1cm, text centered},
  bbox/.style={box, fill=blue!20},
  pbox/.style={box, fill=violet!20},
  gbox/.style={box, fill=green!20},
  ybox/.style={box, fill=yellow!30},
  arrow/.style={-{Latex}, thick}
]


\node (L16) [gbox] {\lemnode{lem: bound on outside option}};
\node (L15) [gbox, left=of L16] {\lemnode{lem: belief upper bound is well-behaved}};

\node (L1) [ybox, above=3cm of L15] {\lemnode{lemma: FEI captures effort incentives}};

\node (L2) [ybox, right=2cm of L1] {\lemnode{lem: failure of FB is uniform}};
\node (L3) [ybox, below=of L2] {\lemnode{lem: continuation value can't keep growing}};

\node (L17) [gbox, below=of L16, opacity=0] {\lemnode{lem: bound on outside option}};

\node (L6) [bbox, right=of L17] {\lemnode{lem: learning or firing}};
\node (L7) [bbox, below=of L6] {\lemnode{lem: opportunistic types are replaced}};
\node (L4) [bbox, above left=of L6, xshift=2cm] {\lemnode{lem: firing hazard bounded below}};
\node (L5) [bbox, above right=of L6, xshift=-2cm] {\lemnode{lem: martingale convergence}};

\node (L10) [bbox, right=of L3, xshift=3cm] {\lemnode{lem: beliefs reach arbitrarily high on path}};

\node (L9) [bbox] at (L5 |- L2) {\lemnode{lem: good news doesn't uniformly cause firing}};

\node (L11) [bbox, below=of L10] {\lemnode{lem: good types can survive}};
\node (L8)  [bbox, above right=of L11] {\lemnode{lem: high reputation gives chance of tenure}};

\node (P3) [bbox, below=of L7] {\propnode{prop: differential replacement}};
\node (T1) [box, below=of P3] {\thmnode{thm: equivalence theorem}};

\node (L12) [pbox, right=of L11] {\lemnode{lem: cutoff rule for first-best equilibrium}};
\node (L13) [pbox, below left=of L12, xshift=2cm] {\lemnode{lem: first best equilibrium}};
\node (L14) [pbox, below right=of L12, xshift=-2cm] {\lemnode{lem: bad equilibrium construction}};


\node (P1) [box, below=7cm of L12] {\propnode{prop: eqm exists}};

\node (P4) [gbox] at (L15 |- L17) {\propnode{prop: outside option can be super low}};

\node (P2) [box] at (L17 |- L7) {\propnode{prop: firing happens with a good reputation}};
\node (C1) [box] at (L17 |- P3) {\cornode{cor: long run voter welfare}};
\node (C2) [box] at (L17 |- T1) {\cornode{cor: equivalence theorem with extra condition}};


\draw[arrow] (L2) -- (L3);
\draw[arrow] (L2) -- (L10);
\draw[arrow] (L3) -- (L16);
\draw[arrow] (L15) -- (L16);

\draw[arrow] (L16) -- (P4);

\draw[arrow] (L7) -- (P2);

\draw[arrow] (L3) -- (L10);
\draw[arrow] (L9) -- (L10);
\draw[arrow] (L10) -- (L11);
\draw[arrow] (L8) -- (L11);

\draw[arrow] (L11) to[out=255, in=55] (P3);
\draw[arrow] (L7) -- (P3);
\draw[arrow] (L6) -- (L7);
\draw[arrow] (L4) -- (L6);
\draw[arrow] (L5) -- (L6);
\draw[arrow] (L4) to[out=265, in=135] (L7);

\draw[arrow] (P3) -- (T1);

\draw[arrow]
  (L1.south west)
    to[out=240, in=235, looseness=1.1]
      (T1.south);

\draw[arrow]
  (L1)
    to[out=42, in=25, looseness=1.2]
      ($(L11.north)!0.5!(L11.north east)$);

\draw[arrow] (P3) -- (C1);

\draw[arrow]
  (L16.south)
    to[out=225, in=90]
      ([xshift=-1cm]C1.north west)
    to[out=-90, in=135]
      (C2.north west);

\draw[arrow] (L12) -- (L13);
\draw[arrow] (L12) -- (L14);
\draw[arrow] (L13) -- (T1);
\draw[arrow] (L14) to[out=235, in=25] (T1);

\end{tikzpicture}
}

\vspace{-9em}
\begin{center}
\begin{tikzpicture}[
  node distance=1cm and 1.5cm,
  box/.style={rectangle, draw, rounded corners, minimum width=1.5cm, minimum height=1cm, text centered},
  bbox/.style={box, fill=blue!20},
  pbox/.style={box, fill=violet!20},
  gbox/.style={box, fill=green!20},
  ybox/.style={box, fill=yellow!30},
  arrow/.style={-{Latex}, thick}
]
\node[ybox] (y) {\hyperref[sec: proof preliminaries]{Preliminary calculations}};
\node[bbox, right=0.4cm of y] (b) {\hyperref[sec: proof EFE]{Eventual full effort}};
\node[pbox, right=0.4cm of b] (p) {\hyperref[sec: proof constructions]{Equilibrium constructions}};
\node[gbox, right=0.4cm of p] (g) {\hyperref[sec: proof outside option]{Outside option}};
\end{tikzpicture}
\end{center}

\caption{This flow chart depicts the logical relationships between every lemma (L), theorem (T), corollary (C), and proposition (P). 
An arrow pointing from one result to another indicates that the former is a direct input to the latter. Different colors correspond to different sections of the proof appendix.}
\label{fig: lemma flowchart}
\end{figure}

\paragraph{Preliminary calculations.}

\cref{sec: proof preliminaries} begins with some preliminary calculations. Adapting standard repeated-games arguments, \cref{lemma: FEI captures effort incentives} confirms that 
\condref{FEI} exactly captures the ability to incentivize politicians to exert effort whenever in office. Using a compactness argument, \cref{lem: failure of FB is uniform} notes that, if \condref{FEI} fails, it fails uniformly. More specifically, the lemma finds a lower bound on how much a politician's continuation value must be allowed to grow (following some signal realizations) in order to make working incentive compatible. Applying this bound, \cref{lem: continuation value can't keep growing} finds a finite sequence of signals during which the incumbent must strictly prefer to shirk at least once.

\paragraph{Eventual full effort.}

In \cref{sec: proof EFE}, we establish that a failure of \condref{FEI} generates eventual full effort when replacement costs are low. The key mechanism is one of selection (\cref{prop: differential replacement}): opportunistic types are always replaced, whereas good types have a positive probability of being kept forever, meaning the latter is sure to happen eventually. (In \cref{sec: proof extra}, we observe that this selection easily implies a welfare statement---\cref{cor: long run voter welfare}, which says the time-$t$ voter's payoff converges to first best as $t\to\infty$.)

The argument for why opportunistic types are eventually replaced has four steps. First, \cref{lem: firing hazard bounded below} yields a finite time horizon and a lower bound on the probability that an opportunistic incumbent is replaced over that horizon, starting from any history at which he exerts effort. Intuitively, if replacement were arbitrarily unlikely over all long horizons, permanent shirking would be profitable. Full-support monitoring allows us to convert the implied lower bound on the replacement rate under such a deviation to a lower bound on the on-path replacement rate. 

Second, \cref{lem: martingale convergence} shows that any officeholder is eventually replaced, or has his type asymptotically revealed, or asymptotically behaves indistinguishably from the good type. The logic is as in \cite*{CMS04}. By the martingale convergence theorem, reputation converges almost surely. Since monitoring statistically identifies effort, beliefs can have an interior limit with infinitely many signals only if effort is asymptotically uninformative about type—that is, only if the opportunistic type behaves like the good type in the limit.

Third, \Cref{lem: learning or firing} shows that that any given officeholder is in fact eventually replaced or has his type asymptotically revealed. Given \cref{lem: martingale convergence}, we need only rule out that the officeholder is never replaced and yet asymptotically works with probability approaching $1$ even if opportunistic. But in the latter event, effort is eventually incentive compatible in every period, and \cref{lem: firing hazard bounded below} implies the officeholder is eventually replaced. 

Fourth, when the replacement cost is low, \cref{lem: opportunistic types are replaced} completes the argument that an opportunistic type is eventually replaced. Because of \Cref{lem: learning or firing}, we need only prove that an officeholder cannot be retained forever while voters become asymptotically certain that he is opportunistic. We rule that possibility out via two cases. If the opportunistic type is eventually expected to work with positive probability, then \cref{lem: firing hazard bounded below} guarantees eventual replacement. In the complementary case, the officeholder eventually has sufficiently low reputation and is expected to shirk if opportunistic that voters strictly prefer replacement.

We next argue through four lemmas that if \condref{FEI} fails, good types have a positive probability of never being replaced. \Cref{lem: high reputation gives chance of tenure} shows that reaching a sufficiently high on-path reputation—one at which a voter would retain the incumbent even if an opportunistic type were expected to surely shirk—implies a positive probability of permanent retention.\footnote{Importantly, this positive probability is unconditional on type; indeed, \cref{lem: opportunistic types are replaced} implies it is zero conditional on the opportunistic type.} Since beliefs are a martingale, Doob’s optional stopping theorem implies that once reputation exceeds this threshold, it has a positive probability of remaining there forever, in which event the incumbent is never replaced. 

\cref{lem: good news doesn't uniformly cause firing} shows that any on-path career history is followed by some immediately subsequent on-path history with an additional ``good’’ signal (one strictly indicative of effort). The proof considers two cases. If the incumbent surely shirks at the current history, then learning implies that after any good signal the next voter has a strict incentive to retain the incumbent. If the incumbent works with positive probability at the current history, his incentive compatibility implies that he must be rewarded with a positive probability of retention following some good signal.

\cref{lem: beliefs reach arbitrarily high on path} then says an incumbent's reputation can become arbitrarily good at some on-path histories if \condref{FEI} fails. The key to establishing this fact is to consider play at on-path histories with beliefs very close to a supremum that is putatively interior. As \cref{lem: good news doesn't uniformly cause firing} says some good signal results in a positive probability of retention, the 
hypothesis that we are near an interior supremum means opportunistic types work with high probability. But then, there is little learning no matter the signal, so next period's belief is again close to that supremum. Iterating this argument over the (finite) horizon identified by \cref{lem: continuation value can't keep growing} generates a history at which reputation is close to the interior supremum and yet the incumbent shirks, a contradiction. 

\cref{lem: good types can survive} finally concludes that any good officeholder has a positive probability of perpetual retention when \condref{FEI} fails. In light of \cref{lem: beliefs reach arbitrarily high on path} (and because Bayesian updating precludes only opportunists obtaining arbitrarily high reputations), the result is immediate as long as the threshold required in \cref{lem: high reputation gives chance of tenure} is strictly below $1$. It obviously is when the replacement cost is strictly positive. With a $0$ replacement cost, we just need 
that the voters' outside option is strictly less than $1$ in any equilibrium. But an equilibrium with outside option $1$ would have incumbents working at every on-path history, which cannot happen when \condref{FEI} fails, by \cref{lemma: FEI captures effort incentives}.

\paragraph{Equilibrium constructions.}

\cref{sec: proof constructions} constructs two kinds of equilibria when \condref{FEI} holds. By \cref{lemma: FEI captures effort incentives}, \condref{FEI} is equivalent to the existence of a voter strategy (ignoring voter incentives) that induces an officeholder to exert effort in every period until replacement. \cref{lem: cutoff rule for first-best equilibrium} shows that when such a strategy exists, there is one with a simple structure: a time- and history-independent likelihood-ratio test in which all “good” signals (those indicative of effort), and possibly some “bad” signals, pass the test.

The logic underlying the proof of \cref{lem: cutoff rule for first-best equilibrium} is as follows.\footnote{The proof itself is more algebraic, expressed in terms of the function $v$ witnessing \condref{FEI}.} Given any voter strategy that incentivizes full effort---which without loss gives the highest possible continuation value to a new officeholder---we can transform it without weakening incentives. First, rewarding all good signals with the highest continuation value relaxes both contemporaneous and past incentive constraints. Second, a replication argument implies that treating signals with the same likelihood ratio preserves incentives. Third, rearranging continuation values to be monotone in the likelihood ratio, while holding the on-path value fixed, only strengthens incentives. The resulting voter strategy takes the form of a likelihood-ratio cutoff test: replace the incumbent for signals below the cutoff, retain him with constant high probability above it, and (at most) mix between these two behaviors at the cutoff. Finally, we show that incentive constraints are relaxed by eliminating mixing and, more straightforwardly, by lowering the replacement rate after ``passing’’ signals (which amounts to scaling up continuation values). Hence, we can further simplify to the incumbent being replaced only after a failing signal.

With \cref{lem: cutoff rule for first-best equilibrium} in hand, we can explicitly construct some equilibria under \condref{FEI}. \cref{lem: first best equilibrium} constructs an equilibrium that attains full effort. On path, the officeholder always works, and the voter uses the strategy detailed by \cref{lem: cutoff rule for first-best equilibrium}. If any voter was supposed to have replaced a politician but did not, then
this politician always shirks thereafter, her reputation is permanently zero, and all future voters replace him.\footnote{\label{fn:wPBE-use1}When the replacement cost is positive ($c>0$), the belief at the first off-path history in this construction is valid because we are using \emph{weak} PBE. When $c=0$, our construction can be readily modified to be a PBE.} As long as the replacement cost is low, voters find it incentive compatible to replace the incumbent when they should, and every other incentive constraint holds by construction.

The final input to \autoref{thm: equivalence theorem} is \cref{lem: bad equilibrium construction}, which constructs an equilibrium that does not attain eventual full effort. The construction builds on the likelihood ratio test of \cref{lem: cutoff rule for first-best equilibrium}. Call a signal a Pass if it passes the test and a Fail otherwise; all Fail signals (and possibly some Pass signals) are bad in the sense of being strictly indicative of shirking. For a given small enough replacement cost, our construction---illustrated in \autoref{fig: bad equilibrium construction} for the binary-signal case---is as follows.

After the first Pass signal, continuation play proceeds as in \cref{lem: first best equilibrium}'s full-effort construction: on path, the officeholder keeps exerting effort and is replaced the first time he produces a Fail signal. Prior to the first Pass signal, voters replace the incumbent with a fixed probability that makes the opportunistic type indifferent between working and shirking; such a probability exists by (\cref{lem: cutoff rule for first-best equilibrium} and) the definition of Pass and Fail signals. The opportunistic type’s initial effort probability is strictly below $1$, but high enough that for the relevant range of reputations, voters who expect opportunistic incumbents to surely shirk are willing to incur the replacement cost. Before the first Pass signal, effort is determined recursively by keeping voters indifferent between retaining and replacing the incumbent at every history; because Fail signals are bad, the incumbent’s reputation declines over time, which can be exactly offset by increasing the opportunistic type’s effort probability. This construction does not attain eventual full effort because every officeholder is eventually replaced and voters are therefore perpetually exposed to the strictly positive initial shirking probability.\footnote{As mentioned in \autoref{subsec:bad-longrun}, this also implies that average voter welfare is strictly below first best, even gross of the replacement cost.}

\subsubsection{Preliminary Calculations}\label{sec: proof preliminaries}

\cref{table:notation} summarizes some of the key notation used in the formal arguments that follow.

We begin by confirming that \condref{FEI} indeed captures incumbents' incentives to exert full effort. Although this fact follows directly from adapting standard self-generation arguments \citep*{APS90}, we include a proof for completeness.
\begin{Lemma}\label{lemma: FEI captures effort incentives}
The following are equivalent: 
\begin{enumerate}
    \item\label{lemma: FEI captures effort incentives, contract} Some strategy profile  attains full effort and satisfies politician incentive compatibility.
    \item\label{lemma: FEI captures effort incentives, equilibrium} Some equilibrium  of the modified game with $(\prior,\cost)=(0,0)$ attains full effort.
    \item\label{lemma: FEI captures effort incentives, self-generation} \condref{FEI} holds.
\end{enumerate}
\end{Lemma}
\begin{proof}

Let us first show condition~\ref{lemma: FEI captures effort incentives, contract} implies condition~\ref{lemma: FEI captures effort incentives, self-generation}. To that end, suppose the strategy profile $(\pstrat,\vstrat)$ attains full effort and satisfies politician incentive compatibility. Let $\cval:\Hist\to[0,1]$ capture the  politician's continuation values under this strategy profile, and let $\Hist^* \subseteq \Hist$ denote the set of on-path histories.  
Consider the set $$V:=\curlyb{\paren{\cval(h), \ \paren{\brac{1-\vstrat(h,s)}\cval(h,s)}_{s\in S}}:\ h\in\Hist^*}\subseteq [0,1-\ecost]\times[0,1-\ecost]^S,$$ and let $V_0\subseteq[0,1]$ denote the projection of $V$ onto its first coordinate. These sets are nonempty because $\emptyset\in\Hist^*$. So let $\bar\cval_0:=\sup V_0\in[0,1]$. 
By definition, $V$ contains some sequence $\paren{\cval^n_0,(\cval^n(s))_{s\in S}}_{n=1}^\infty$ such that $\lim_{n\to\infty} \cval^n_0= \bar\cval_0$. Dropping to a subsequence, we may assume (because $[0,1]$ is compact) that $\cval^n=(\cval^n(s))_{s\in S}$ also converges to some $\bar\cval\in[0,1]^S$ as $n\to\infty$. Let us show that $\bar\cval$ witnesses \condref{FEI}. Because it is incentive compatible for the politician to work at every history in $\Hist$, every $n$ has $\cval^n_0=(1-\delta)(1-\ecost)+\delta\mon_1\cdot\cval^n\geq(1-\delta)+\delta\mon_0\cdot\cval^n,$
so that taking limits yields $$\bar\cval_0=(1-\delta)(1-\ecost)+\delta\mon_1\cdot\bar\cval\geq(1-\delta)+\delta\mon_0\cdot\bar\cval.$$ Meanwhile, any $h\in\Hist^*$ and $s\in S$ have either $\vstrat(h,s)=1$ or $(h,s)\in\Hist^*$ by definition, hence $\brac{1-\vstrat(h,s)}\cval(h,s)$ is a convex combination of $0$ and something in $V_0$. Thus, every $n\in\N$ and $s\in S$ have $\cval^n(s)\leq\bar\cval_0$, so that taking limits yields $\bar\cval\in[0,\bar\cval_0]^S$. Hence, $\bar\cval_0=\mon_1\cdot\bar\cval\geq \mon_0\cdot\bar\cval$ and $\bar\cval_0\geq\max\bar\cval(S)$, witnessing \condref{FEI}.

To see that condition~\ref{lemma: FEI captures effort incentives, self-generation} implies condition~\ref{lemma: FEI captures effort incentives, equilibrium}, suppose \condref{FEI} holds, witnessed by $\cval:S\to\Rpos$. Letting $\cval_0:=(1-\delta) (1-\ecost) +\delta\sum_{s \in S} \mon_1(s)\cval(s) \geq \max_{s \in S} \cval(s),$ we can construct a full-effort equilibrium as follows. At every history $h$, the politician chooses $\pstrat(h)=1$, the voter's belief is $\belief(h)=0$, and following each signal $s\in S$ the voter chooses $\vstrat(h,s)=1-\tfrac{\cval(s)}{\cval_0}\in[0,1]$. The Bayesian property, and the incentives of the voter (who has payoff $1$ for either re-election or replacement) are trivial. Toward politician incentives, note that the constant function $h\mapsto\cval(h)=\cval_0$ satisfies the recursive equation for politician continuation payoffs. Because a unique such bounded function $\Hist\to\R$ exists (by the Banach fixed point theorem), it must be that the politician's continuation value is $\cval_0$ at every history. But then, \eqref{IC for FB} tells us the politician finds it optimal to work at every history, confirming equilibrium.

Finally, condition~\ref{lemma: FEI captures effort incentives, equilibrium} obviously implies condition~\ref{lemma: FEI captures effort incentives, contract}, delivering the lemma.
\end{proof}

The following lemma says that if \condref{FEI} fails, it fails in a uniform way.

\begin{Lemma}\label{lem: failure of FB is uniform}
If $(\ecost,\delta,\mon)$ do not satisfy \condref{FEI}, then some $T\in\N$ exists such that every $\cval\in\Rpos^S$ satisfying \eqref{IC for FB} has 
$(1-\delta) (1-\ecost) +\delta\sum_{s \in S} \mon_1(s)\cval(s) < \max_{s \in S} \cval(s)-\frac1{T}$.
\end{Lemma}
\begin{proof}
Suppose \condref{FEI} fails.

The set $V$ of $\cval\in[0,1]^S$ satisfying \eqref{IC for FB} is compact, and the function $\cval\mapsto \max \cval(S)- \brac{(1-\delta) (1-\ecost) +\mon_1\cdot\cval}$ is continuous on it. Failure of \condref{FEI} then implies this function is strictly positive on $V$, and so bounded below by a strictly positive number. 

Meanwhile, any $\cval\in\Rpos^S$ with $\bar\cval:=\max\cval(S)\geq 1$ has 
$$(1-\delta) (1-\ecost) +\delta\mon_1\cdot\cval\leq (1-\delta) (\bar\cval-\ecost) + \delta \bar\cval=\bar\cval - (1-\delta)\ecost,$$
so that $$\inf_{v\in \Rpos^S:\ \cval \text{ satisfies } \eqref{IC for FB}} \curlyb{ \max\cval(S)- \brac{(1-\delta) (1-\ecost) +\delta\mon_1\cdot\cval}}>0.$$
Thus $\tfrac1{T}$ is strictly below this 
 infimum value for $T$ large enough.
\end{proof}

The next lemma finds, when \condref{FEI} fails, a sequence of signals and a duration over which an incumbent cannot be incentivized to work.
\begin{Lemma}\label{lem: continuation value can't keep growing}
Suppose \condref{FEI} fails, let $T$ be as delivered by \cref{lem: failure of FB is uniform}, and fix any equilibrium and any $h_0\in\Hist$ with $\cval(h_0)>0$. 
Recursively define $s_0,\ldots,s_{T-1}\in S$ and $h_1,\ldots,h_{T}\in \Hist$ by letting
$s_t\in\arg\max_{s\in S} \curlyb{\brac{1-\vstrat(h_t,s)}\cval(h_t,s)}$ and $h_{t+1}=(h_{t},s_{t})$
for $t\in\{0,\ldots,T-1\}$.
Then some $t\in\{0,\ldots,T-1\}$ has $\pstrat(h_t)=0$.
\end{Lemma}
\begin{proof}
Assume for a contradiction that every $t\in\{0,\ldots,T-1\}$ has $\pstrat(h_t)>0$. For each such $t$, the fact that the incumbent willingly exert effort tells us the vector $\cval^t:=\paren{\brac{1-\vstrat(h_t,s)}\cval(h_t,s)}_{s\in S} \in [0,1]^S$ satisfies \eqref{IC for FB}. Hence, \cref{lem: failure of FB is uniform}
says $$\cval(h_t)=(1-\delta) (1-\ecost) +\delta\sum_{s \in S} \mon_1(s)\cval^t(s) < \cval^t(s_t)-\frac1{T}\leq\cval(h_{t+1})-\frac1{T}.$$ Therefore, 
$$0\leq \cval(h_0)< \cval(h_T)-T\frac1{T}\leq 1 -T\frac1{T} = 0,$$
a contradiction.
\end{proof}

For the analysis that follows, we introduce some notation to describe voters' belief updating.

\begin{notation}
For any belief $\belief\in(0,1]$, work probability $a\in[0,1]$ for the opportunistic type, and signal $s\in S$, let $$\bupdate_a(\belief|s):=\frac{\belief\mon_1(s)}{\brac{\belief+(1-\belief)a}\mon_1(s) + (1-\belief)(1-a)\mon_0(s)}\in(0,1]$$
denote the updated Bayesian belief that a voter would have after observing signal $s$.

For any $\belief\in(0,1]$ and bound $\pbound\in[0,1]$, let $$\bbound_\pbound(\belief):=\sup_{a\in[\pbound,1],\ s\in S} \bupdate_a(\belief|s)\in(0,1]$$
denote the highest level beliefs can grow to in a period when the opportunistic-type incumbent's expected effort is at least $\pbound$.

For any $t\in\mathbb Z$ and $\pbound\in[0,1]$, let $\bbound_\pbound^t:(0,1]\to(0,1]$ denote the $t$-fold composition of $\bbound_\pbound$.
\end{notation}

\subsubsection{Eventual Full Effort}\label{sec: proof EFE}

Primarily for the argument that a failure of \condref{FEI} implies all equilibria attain eventual full effort, we will also use the following notation.

\begin{notation}
Given any equilibrium, we define the following random variables. 

For any $t\in\Zpos$, let $\replacerv_t\in\{0,1\}$ be the random variable indicating whether the first officeholder has been replaced as of the beginning of period $t$; and let $\replacetimerv:=\inf\{t\in\Zpos:\ \replacerv_t=1\}\in\Zpos\cup\{\infty\}$ denote the first replacement time. 
For any $t\in\Zpos$ with $t\leq\replacetimerv$, let $\histrv_t\in\Hist_t$ be the random variable denoting the first officeholder's public career history up to time $t$. For any $t\in\Zpos$, let $\beliefrv_t:=\belief(\histrv_{t\wedge\replacetimerv})\in[0,1]$ denote the first officeholder's reputation, and let $\actprv_t:=(1-\replacerv_t)\pstrat(\histrv_{t\wedge\replacetimerv})\in[0,1]$ denote the expected effort of the initial officeholder conditional on being opportunistic. Let $\replacerv_\infty$ [resp., $\beliefrv_\infty$, $\actprv_\infty$] denote the limit as $t\to\infty$ of $\replacerv_t$ [resp., $\beliefrv_t$, $\actprv_t$] in the event that this limit exists.

Any equilibrium induces a joint law over the stochastic process $(\actrv_t,\beliefrv_t,\replacerv_t)_{t=0}^\infty$. We take all probabilistic statements to be with respect to this law unless stated otherwise (and we suppress the underlying probability space and the specific equilibrium in our notation). 
\end{notation}

For ease of notation, we work throughout with the above random variables concerning only the first officeholder. Given our focus on personal symmetric equilibria, results we prove about the first officeholder apply without change to any officeholder.

We now show that politicians' incentives imply a lower bound on the opportunistic type's (time-averaged) hazard rate of being replaced.
\begin{Lemma}\label{lem: firing hazard bounded below}
Some $T\in\N$ and $\pbound\in(0,1)$ exist such that, in any equilibrium and at any history $h\in\Hist$ with $\pstrat(h)>0$, an opportunistic-type officeholder has a probability of at least $\pbound$ of being replaced in the next $T$ periods. 
\end{Lemma}
\begin{proof}
Take some large enough $T\in\N$ and small enough $\varepsilon\in(0,1)$ to ensure that 
\begin{equation}
    (1-\varepsilon)(1-\delta^T)>(1-\delta)(1-\ecost)+\delta. \label{e:firing harzard inequality}
\end{equation}
Denoting $S_0:=\mathrm{supp}(\mon_0)$, let $\lrat:=\min_{s\in S_0} \tfrac{\mon_1(s)}{\mon_0(s)}\in(0,1)$ and $\pbound:=\lrat^T\varepsilon$.

First, observe that at any history $h\in\Hist$ with $\pstrat(h)>0$, incentive compatibility requires that an opportunistic incumbent has a probability at least $\varepsilon$ of being replaced in the next $T$ periods, conditional on always shirking starting from $h$. The reason is that if this probability of replacement were strictly less than $\epsilon$, then the continuation payoff from deviating to always shirking would exceed $(1-\epsilon)(1-\delta^T)$, which by \eqref{e:firing harzard inequality} strictly exceeds the maximum feasible continuation payoff from working at $h$.

Next, observe that for any sequence of $T$ signals, its probability of arising (conditional on the incumbent not being replaced) under any politician strategy is at least $\lrat^T$ times as high as its probability of arising if the politician were to always shirk. Therefore, the incumbent at $h$ will be replaced in the next $T$ periods with probability at least $\lrat^T\varepsilon=\pbound$.
\end{proof}

We now record a lemma that says either an officeholder's type will be revealed or learning will halt. The reasoning behind this lemma is essentially the same as that of \citet*{CMS04}.

\begin{Lemma}\label{lem: martingale convergence}
In any equilibrium, the stochastic process $(\replacerv_t,\actprv_t,\beliefrv_t)_{t=0}^\infty$ satisfy the following: \begin{enumerate}
    \item The sequence $(\replacerv_t,\beliefrv_t)_t$ converges a.s.~to a random variable $(\replacerv_\infty,\beliefrv_\infty)$.
    \item We a.s.~have $\replacerv_\infty=1$ (that is, $\replacetimerv<\infty$) or $\beliefrv_\infty\in\{0,1\}$ or $(\actprv_t)_t$ converges to 1.
\end{enumerate}
\end{Lemma}
\begin{proof}
The first point follows from pathwise monotonicity of $\replacerv$ and the martingale convergence theorem for the bounded martingale $\beliefrv$. 

For the second point, consider the event in which $\replacerv_\infty\neq1$ and $\beliefrv_\infty\notin\{0,1\}$---so $\replacerv_\infty=0$ (that is, $\replacetimerv=\infty$) and $0<\beliefrv_\infty<1$.  For any $\belief\in(0,1)$, we have $\min_{s\in S,\ \tilde a\in[0,1]}\mon_{\belief+(1-\belief)\tilde a}(s)>0$; and some $s\in S$ has $\mon_1(s)\neq\mon_0(s)$. Therefore, when a belief lies in some small enough open neighborhood of $\belief$ and the officeholder's expected effort is bounded away from $1$, the probability of the belief exiting said neighborhood is bounded below. Thus, in the event that $0<\beliefrv_\infty<1$, we must have $\actprv_t\to1$. 
\end{proof}

We now prove that in any equilibrium, an officeholder is either eventually replaced or his type is asymptotically revealed.

\begin{Lemma}\label{lem: learning or firing}
Consider any equilibrium, and let $(\replacerv_\infty,\beliefrv_\infty)$ be as in the statement of \cref{lem: martingale convergence}. We a.s.~have $\replacerv_\infty=1$ (that is, $\replacetimerv<\infty$) or $\beliefrv_\infty\in\{0,1\}$.
\end{Lemma}
\begin{proof}
Consider the event in which $\beliefrv_\infty>0$ and $\actprv_\infty:=\lim_{t\to\infty}\actprv_t=1$. In this event, we in particular have $\actprv_t>0$ for sufficiently large $t\in\Zpos$. \cref{lem: firing hazard bounded below} delivers some $T\in\N$ and $\pbound\in(0,1)$ such that the probability of replacement in the next $T$ periods, conditional on the opportunistic-type officeholder, is eventually at least $\pbound$ in this event. But then, because $\actprv_\infty=1$, the good type also eventually has at least a $\tfrac{\pbound}{2}$ probability of replacement in the next $T$ periods. Hence, by the Borel-Cantelli Lemma, with probability $1$ conditional on this event, we have $\replacerv_\infty=1$.

We have argued that, excluding a null event, if $\beliefrv_\infty>0$ and $\actprv_\infty=1$, then $\replacerv_\infty=1$. Combining this observation with \cref{lem: martingale convergence} delivers the lemma.
\end{proof}

The next lemma says that, when the replacement cost is small, an opportunistic officeholder is eventually replaced.

\begin{Lemma}\label{lem: opportunistic types are replaced}
If $\cost<\prior$, then with probability $1$ conditional on being an opportunistic type, an officeholder is replaced.
\end{Lemma}
\begin{proof}
Suppose $\cost<\prior$. Let us first show that (excepting a zero probability event) $\replacerv_\infty=1$ if $\beliefrv_\infty=0$. Indeed, if $\beliefrv_\infty=0$, then some $\tau\in\Zpos$ exists such that $\beliefrv_t<\min\{\cost,\tfrac12\}$ for every $t\geq\tau$. Let us now show $\replacerv_\infty=1$ in two different cases. First, if $\actprv_t=0$ for some $t\geq\tau$, then voter incentives require $\replacerv_t=1$ so that $\replacerv_\infty=1$ too. Second, if $\actprv_t>0$ for every $t\geq\tau$, then \cref{lem: firing hazard bounded below} delivers some $T\in\N$ and $\pbound\in(0,1)$ such that the probability of replacement in the next $T$ periods is at least $\tfrac{\pbound}{2}$ from time $\tau$ onward. Then, by the Borel-Cantelli Lemma, we have $\replacerv_\infty=1$.

\cref{lem: learning or firing} tells us $\replacerv_\infty=1$ or $\beliefrv_\infty\in\{0,1\}$ a.s., and the argument above rules out the case in which $\replacerv_\infty=\beliefrv_\infty=0$ with positive probability. Therefore, we a.s.~have $\replacerv_\infty=1$ or $\beliefrv_\infty=1$. The lemma follows because the probability of $\beliefrv_\infty=1$ conditional on the officeholder being an opportunistic type is necessarily zero.
\end{proof}

The next lemma shows that if an officeholder attains a high enough reputation on-path that the voter would currently retain him irrespective of opportunistic-type behavior, then he has a positive probability (unconditional on his type) of never being replaced.

\begin{Lemma}\label{lem: high reputation gives chance of tenure}
Fix an equilibrium, and let $\ooption$ denote the voters' outside option.
For any $t\in\Zpos$ and history $h\in\Hist$ is such that $\Prob\curlyb{\histrv_t=h \text{ and }\replacerv_t=0}>0$ and $\belief(h)>\ooption-\cost$, we have $\Prob\curlyb{\histrv_t=h,\ \beliefrv_\infty\geq\belief(h), \text{ and }\replacerv_\infty=0}>0$.
\end{Lemma}
\begin{proof}
Define the stopping time $\boldsymbol{\tau}:=\inf\curlyb{\tau\in\Zpos:\ \tau\geq t \text{ and }\beliefrv_\tau\leq\ooption-\cost}$. Note that the random variable $\beliefrv_{\boldsymbol{\tau}}$ is well defined a.s.~(even if $\boldsymbol{\tau}=\infty$) by the martingale convergence theorem. Doob's optional stopping theorem implies $\E[\beliefrv_{\boldsymbol{\tau}}\ |\ \histrv_t=h]=\belief(h)$. Therefore, $\Prob\curlyb{\beliefrv_{\boldsymbol{\tau}}\geq \belief(h) \ |\ \histrv_t=h}>0.$
But if $\histrv_t=h$ and $\beliefrv_{\boldsymbol{\tau}}\geq \belief(h)>\ooption-\cost$, then the voters' beliefs remain in $(\ooption-\cost,1]$ from time $t$ onward---and hence, by voter incentives, they do not replace the officeholder from time $t$ onward. 
So $$\Prob\curlyb{\replacerv_\infty=0 \text{ and } \beliefrv_\infty\geq\belief(h)\ |\ \histrv_t=h \text{ and } \replacerv_t=0}\geq
\Prob\curlyb{\beliefrv_{\boldsymbol{\tau}}\geq \belief(h)\ |\ \histrv_t=h \text{ and } \replacerv_t=0}
>0.$$ 
Because $\Prob\curlyb{\histrv_t=h \text{ and } \replacerv_t=0}>0$ by hypothesis, the result follows.
\end{proof}

The following lemma says that, if a career history is consistent with an officeholder not having been replaced yet, then some good news (that is, some signal indicative of effort) exists that does not lead to immediate firing.

\begin{Lemma}\label{lem: good news doesn't uniformly cause firing}
In any equilibrium, any on-path history $h$ admits some $s\in S$ with $\mon_1(s)>\mon_0(s)$ and $(h,s)$ on path.
\end{Lemma}
\begin{proof}
Let $\Hist^*$ denote the set of on-path histories: $h\in\Hist$ if and only if some $t\in\Zpos$ has $\Prob\curlyb{\histrv_t=h \text{ and }\replacerv_t=0}>0$.
First, consider the case in which $\pstrat(h)=0$. That $h\in\Hist^*$ while $\pstrat(h)=0$ implies $\belief(h)\geq\ooption-\cost$, where $\ooption$ is the voters' outside option. Indeed, this inequality holds by voter incentives at $h$ if $h\neq\emptyset$, and it holds by definition (since $\cost\geq0$) for $h=\emptyset$. In this case, we can use any $s\in S$ with $\mon_1(s)>\mon_0(s)$---which exists because $\mon$ is informative.
The Bayesian property implies (since $\belief(\cdot)>0$ at on-path histories because $\mon_1$ has full support) that $\belief(h,s)>\belief(h)$. Hence $\belief(h,s)+\brac{1-\belief(h,s)}\pstrat(h,s)>\ooption-\cost$, and so voters have a strict incentive to keep the incumbent at history $(h,s)$. Thus, $(h,s)\in\Hist^*$.

Now, turn to the case in which $\pstrat(h)>0$. The fact that the incumbent optimally chooses to work at $h$ implies \begin{eqnarray*}
0 
&\leq& \curlyb{(1-\delta)(1-\ecost ) + \delta\sum_{s\in S}\mon_1(s)\brac{1-\vstrat(h,s)}\cval(h,s)}\\
&&- 
\curlyb{(1-\delta)1 + \delta\sum_{s\in S}\mon_0(s)\brac{1-\vstrat(h,s)}\cval(h,s)} \\
&=& -(1-\delta)\ecost+ \delta\sum_{s\in S}\brac{\mon_1(s)-\mon_0(s)}\brac{1-\vstrat(h,s)}\cval(h,s)\\
&<& \sum_{s\in S:\ \mon_1(s)>\mon_0(s)}\delta\brac{\mon_1(s)-\mon_0(s)}\brac{1-\vstrat(h,s)}\cval(h,s).
\end{eqnarray*}
Therefore, at least one summand is strictly positive. In particular, some $s\in S$ has $\mon_1(s)>\mon_0(s)$ and $\vstrat(h,s)<1$. Because $h\in\Hist^*$, the latter property tells us $(h,s)\in\Hist^*$ too.
\end{proof}

The next lemma says that when \condref{FEI} fails, on-path equilibrium reputation can become arbitrarily optimistic.

\begin{Lemma}\label{lem: beliefs reach arbitrarily high on path}
Suppose \condref{FEI} fails, and fix any equilibrium. Any $\hat\belief\in(0,1)$ has $\Prob\curlyb{\beliefrv_t>\hat\belief \text{ and } \replacerv_t=0}>0$ for some $t\in\Zpos$.
\end{Lemma}
\begin{proof}
We begin with some useful definitions and notation. 
Let $\Hist^*$ be the set on-path histories (as also defined at the start of the proof of \cref{lem: good news doesn't uniformly cause firing}), let $\Pi^*:=\{\belief(h):\ h\in\Hist^*\}$, and let $\bar\belief:=\sup(\Pi^*)\in[\prior,1]$. 
Define $S_1:=\{s\in S:\ \mon_1(s)>\mon_0(s)\}$, and note that the function $(a,\belief)\mapsto\min_{s\in S_1}\bupdate_a(\belief|s)$ is continuous in on $[0,1]\times(0,1)$ because each of the finitely many maximands is. Hence, because $[0,1]$ is compact and the function above takes $(1,\belief)\mapsto\belief$ for $\belief\in(0,1)$, the function $\underline{a}:(0,1)\to[0,1]$ given by $$\underline{a}(\belief):=\min\curlyb{a\in[0,1]:\ \min_{s\in S_1}\bupdate_a(\belief|s)\leq\bar\belief}$$
is the lower envelope of a compact-valued and upper hemicontinuous correspondence---hence is well-defined and lower semicontinuous.

Our goal is to prove $\bar\belief=1$; assume otherwise ($0<\bar\belief<1$) for a contradiction. Given this hypothesis, observe that any $a\in[0,1)$ and $s\in S_1$ have $\bupdate_a(\bar\belief|s)>\bar\belief$; thus, $\underline{a}(\bar\belief)=1$. This fact has two consequences. First, lower semicontinuity of $\underline{a}$ tells us $\lim_{\belief\to\bar\belief}\underline{a}(\belief)=1$ too. Second, consequently, some $\underline\belief_0\in(0,\bar\belief)$ has $\underline{a}(\underline\belief_0)>0$ and so $\min_{s\in S_1}\bupdate_0(\underline\belief_0|s)>\bar\belief$.

Now, let $T$ be as given by \cref{lem: failure of FB is uniform}.  
Recursively for $t\in\{1,\ldots,T\}$, let us find some $\underline\belief_{t}\in(\underline\belief_{t-1},\bar\belief)$ with the property that: Every $a\in[0,1]$ with $\inf_{\belief\in(\underline\belief_{t},\bar\belief),\ s\in S_1}\bupdate_a(\belief|s)\leq\bar\belief$ has $\inf_{\belief\in(\underline\belief_{t},\bar\belief),\ s\in S}\bupdate_a(\belief|s)\geq\underline\belief_{t-1}$. 
To do so, observe that $\lim_{\belief\to\bar\belief}\underline{a}(\belief)=1$, and so (by continuity) every $s\in S$ has 
$\bupdate_{\underline a(\belief)}(\belief|s)\to
\bupdate_{1}(\bar\belief|s)=\bar\belief$ as $\belief\to\bar\belief$. Hence, we can set $\underline\belief_{t}\in(\underline\belief_{t-1},\bar\belief)$ close enough to $\bar\belief$ to ensure $\bupdate_{\underline a(\underline\belief_{t})}(\underline\belief_{t}|s)>\underline\belief_{t-1}$ for every $s\in S$. Let us see that such a $\underline\belief_{t}$ has the desired property.
Indeed, any $a\in[0,1]$ with $\inf_{\belief\in(\underline\belief_{t},\bar\belief),\ s\in S_1}\bupdate_a(\belief|s)\leq\bar\belief$ will satisfy $\min_{s\in S_1}\bupdate_a(\underline\belief_{t}|s)\leq\bar\belief$ (by continuity and monotonicity of $\bupdate_a(\cdot|s)$), so that $a\geq\underline a(\underline\belief_{t})$. Hence, any $s\in S$ satisfies 
$\bupdate_a(\underline\belief_{t}|s) \in \mathrm{co}\curlyb{\bupdate_{\underline a(\underline\belief_{t})}(\underline\belief_{t}|s), \ \bupdate_1(\underline\belief_{t}|s)}=\brac{\bupdate_{\underline a(\underline\belief_{t})}(\underline\belief_{t}|s), \ \underline\belief_{t}}.$
In particular, $\bupdate_a(\underline\belief_{t}|s)>\underline\belief_{t-1}$. Then, since $s$ was arbitrary in this calculation and $\bupdate_a(\cdot|s)$ is increasing for every $a\in[0,1]$ and $s\in S$, it follows that $\inf_{\belief\in(\underline\belief_{t},\bar\belief),\ s\in S}\bupdate_a(\belief|s)>\underline\belief_{t-1}$. 

Let us turn now to our equilibrium. Let $h_0\in\Hist^*$ be some on-path history such that $\underline\belief_T<\belief(h_0)<\bar\belief$. Define $s_0,\ldots,s_{T-1}\in S$ and $h_1,\ldots,h_{T}\in \Hist$ as defined in the statement of \cref{lem: continuation value can't keep growing}. That lemma tells us some $\tau\in\{0,\ldots,T-1\}$ has $\pstrat(h_\tau)=0$; consider the smallest such $\tau$. For every $t\in\{0,\ldots,\tau-1\}$, the incumbent finds it optimal to exert effort at history $h_t$, which implies $\brac{1-\vstrat(h_t,s)}\cval(h_t,s)>0$ for some $s\in S$, hence for $s=s_t$. In particular, $\vstrat(h_{t+1})<1$. By induction on $t$, it then follows that $h_1,\ldots,h_{\tau}\in \Hist^*$.  

Next, we show that $\underline\belief_{T-t}<\belief(h_t)<\bar\belief$.
for every $t\in\{0,\ldots,\tau\}$. We will prove this by induction; the base case of $t=0$ holds by fiat. Toward the inductive step, suppose $t\in\{0,\ldots,\tau-1\}$ has $\underline\belief_{T-t}<\belief(h_t)<\bar\belief$. Because $h_t\in\Hist^*$, \cref{lem: good news doesn't uniformly cause firing} tells us some $s'_t\in S_1$ has $(h_t,s_t')\in\Hist^*$ too. But then, by definition of $\bar\belief$, it must be that $\belief(h_t,s_t')\leq\bar\belief$. It then follows from the defining property of $\underline\belief_{T-t}$ that $$\belief(h_{t+1})=\belief(h_t,s_t)>\underline\belief_{T-t-1}=\underline\belief_{T-(t+1)},$$
completing the inductive step.

Finally, we now know that $h_\tau\in\Hist^*$, that $\belief(h_\tau)\in(\underline\belief_{T-\tau},\bar\belief)\subseteq(\underline\belief_0,\bar\belief)$ and that $\pstrat(h_\tau)=0$. \cref{lem: good news doesn't uniformly cause firing} says some $s\in S_1$ has $(h_\tau,s)\in\Hist^*$, and the definition of $\underline\belief_0$ means $\belief(h_\tau,s)>\bar\belief$, which contradicts the definition of $\bar\belief$.
\end{proof}

Now we prove that when \condref{FEI} fails, any good-type officeholder has a positive probability of never being replaced.

\begin{Lemma}\label{lem: good types can survive}
Suppose \condref{FEI} fails. In any equilibrium, the probability of any given good-type officeholder never being replaced is strictly positive.
\end{Lemma}
\begin{proof}
Let us first show the result holds assuming the voters' outside option $\ooption\in[0,1]$ satisfies $\ooption-\cost<1$. Then we will verify this inequality holds.

By symmetry, we need only show the result for the first officeholder. By \cref{lem: beliefs reach arbitrarily high on path} (and since $\ooption-\cost<1$), there is a positive probability of reaching a history at which the first officeholder has not been replaced, and his reputation is strictly positive and strictly greater than $\ooption-\cost$. But then \cref{lem: high reputation gives chance of tenure} says the probability of never being replaced and having $\beliefrv_\infty>0$ is strictly positive too. Hence, the probability $\E\brac{(1-\replacerv_\infty)\beliefrv_\infty}$ of being a good type and never being replaced is strictly positive.  Consequently, the probability of never being replaced \emph{conditional on} being a good type is also strictly positive.

All that remains is to show $\ooption-\cost<1$. The inequality is trivial if either $\ooption<1$ or $\cost>0$, so assume towards contradiction that $\ooption=1$ and $\cost=0$. 
Voter incentives then require $\pstrat(h)=1$ at any history $h\in\Hist$ with $\belief(h)<1$ and $\vstrat(h)<1$. Because the Bayesian property then implies $\belief(h)=\prior<1$ at any history that is reached with positive probability, the opportunistic-type officeholder works at every on-path history. But this contradicts politician incentives, by \cref{lemma: FEI captures effort incentives}.
\end{proof}

Now we argue that when \condref{FEI} fails and replacement costs are low, a good-type politician is eventually retained in office forever.

\begin{proof}[Proof of \cref{prop: differential replacement}]
The first part (which relies on $\cost<\prior$) is exactly \cref{lem: opportunistic types are replaced}, and the second (which does not require our bound on replacement cost) is \cref{lem: good types can survive}.

Now, toward the proposition's last sentence (which again relies on $\cost<\prior$), suppose \condref{FEI} fails. We know every opportunistic officeholder is replaced with probability $1$, whereas every good-type officeholder has a positive (and by symmetry, the same positive) probability of never being replaced. By the Borel-Cantelli Lemma, with probability $1$, some officeholder is never replaced, and this officeholder is a good type.
\end{proof}

\subsubsection{Equilibrium Constructions}\label{sec: proof constructions}

The next lemma says \condref{FEI} is equivalent to a more stringent condition that imposes further restrictions on the continuation values. In terms of induced behavior, if full effort can be incentivized at all, it can be done with a stationary likelihood-ratio test in which sufficiently bad news results in certain replacement.
\begin{Lemma}\label{lem: cutoff rule for first-best equilibrium}
If \condref{FEI} holds, then it can be witnessed by $\cval\in\Rpos^S$ such that:
\begin{itemize}
    \item Some $\bar\cval\in(0,1-\ecost)$ exists such that $\cval(S)=\curlyb{\cval(s):\ s\in S}=\{0,\bar\cval\}$.
    \item Every $s_0$ with $\mon_0(s_0)\leq\mon_1(s_0)$ has $\cval(s_0)=\bar\cval$.
    \item Every $s_0,s_1\in S$ with $\cval(s_0)=0$ and $\cval(s_1)=\bar\cval$ has $\frac{\mon_0(s_0)}{\mon_1(s_0)}>\frac{\mon_0(s_1)}{\mon_1(s_1)}$.
    \item Inequality \eqref{PK for FB} holds with equality.
\end{itemize}
\end{Lemma}
\begin{proof}
Let $S_0:=\curlyb{s\in S:\ \mon_1(s)<\mon_0(s)}$, and let $\succsim$ be the complete transitive binary relation on $S$ given by $s'\succsim s \iff\frac{\mon_0(s')}{\mon_1(s')}\leq\frac{\mon_0(s)}{\mon_1(s)}$. Suppose \condref{FEI} holds and is witnessed by $\cval$. We will successively show that various properties of $\cval$ are without loss of generality. To that end, let $\bar\cval:=\max\cval(S)\in\Rpos$. Note that $0<\bar\cval<1-\ecost$, since any $\cval$ with all entries the same would violate \eqref{IC for FB}, any $\cval$ with $\max\cval(S)\geq1-\ecost$ would violate \eqref{PK for FB}.

First, raising $\cval(s)$ to $\bar\cval$ for each $s\in S\setminus S_0$ relaxes both constraints. So we may assume $\cval(s)=\bar\cval$ for every $s\in S\setminus S_0$. Second, if $\hat S\subseteq S_0$ is some $\sim$-equivalence class, then replacing $\cval(s)$ with $\tfrac{\sum_{\hat s\in\hat S}\mon_1(\hat s)\cval(\hat s)}{\sum_{\hat s\in\hat S}\mon_1(\hat s)}$ for every $s\in\hat S$ preserves both constraints. So we may assume $\cval$ is constant on every $\sim$-equivalence class.

Next, observe that---because $\cval\in[0,\bar\cval]^S$---some $s^*\in S_0$ exists such that $$\sum_{s\succ s^*}\mon_1(s)\bar\cval\leq \sum_{s\in S}\mon_1(s)\cval(s)\leq
\sum_{s\succsim s^*}\mon_1(s)\bar\cval;$$
and so some $\hat\cval\in[0,\bar\cval]$ exists such that 
$$\brac{\sum_{s\succ s^*}\mon_1(s)}\bar\cval
+
\brac{\sum_{s\sim s^*}\mon_1(s)}\hat\cval
=
\sum_{s\in S}\mon_1(s)\cval(s).
$$
Observe now that both inequalities in \condref{FEI} will still be satisfied if we replace $\cval$ with the vector whose $s$ entry is $\bar\cval$ for $s\succ s^*$, is $\hat\cval\in[0,\bar\cval]$ for $s\sim s^*$, and is $0$ for $s\prec s^*$. Indeed, this modification has no effect on \eqref{PK for FB} but relaxes \eqref{IC for FB}. (Moreover, because we had $\cval|_{S\setminus S_0}=\cval$ before the rearrangement and $S\setminus S_0 \succ S_0$, it follows that we still have this property.) So we may further assume $\cval$ takes this form. 

Now, define $\hat z:= 1-\delta\sum_{\tilde s\succ s^* } \mon_1(\tilde s)>0$ and $\bar z:= -\delta\sum_{\tilde s\sim s^* } \mon_1(\tilde s)<0$, and define
$z\in\R^S$ by letting $$z(s):=\begin{cases}
0&:\ s\prec s^* \\
\hat z &:\ s\sim s^* \\
\bar z &:\ s\succ s^* \\
\end{cases}$$
for each $s\in S$. Observe that the set of $\varepsilon\in\R$ with $\bar\cval+\varepsilon\bar z \geq \hat\cval+\varepsilon\hat z\geq 0$ is a closed interval containing zero. It is also bounded: the first inequality is violated as $\varepsilon \to \infty$, and the second is violated as $\varepsilon\to -\infty$. This interval therefore has two extreme points, and at least one of these $\varepsilon$ is such that $\cval+\varepsilon z$ satisfies \eqref{IC for FB} because the latter is affine. Finally, simple algebra shows that this $\cval+\varepsilon z$ satisfies \eqref{PK for FB} too because $\cval$ does. So replacing $\cval$ with this modification, we may assume that $\hat\cval\in\{0,\bar\cval\}$, or equivalently that $\cval(S)\subseteq\{0,\bar\cval\}$. And in fact, we have $\cval(S)=\{0,\bar\cval\}$ because any $\cval$ with all entries the same violates \eqref{IC for FB}.

Finally, because scaling up a vector $\cval$ that satisfies \eqref{IC for FB} relaxes this constraint but tightens \eqref{PK for FB}, we can scale it up to ensure \eqref{PK for FB} holds with inequality.
\end{proof}

We now prove that under \condref{FEI} and a small replacement cost, some equilibrium attains full effort. 
\begin{Lemma}\label{lem: first best equilibrium}
If $\cost\leq1-\prior$ and \condref{FEI} holds, then some equilibrium attains full effort.
\end{Lemma}
\begin{proof}
By \cref{lem: cutoff rule for first-best equilibrium}, some $\bar\cval\in(0,1)$ and $\lrat\geq1$ exist such that $S^*:=\curlyb{s\in S:\ \tfrac{\mon_0(s)}{\mon_1(s)}\leq\lrat}$ satisfies
$$
\bar\cval=(1-\delta)(1-\ecost)+\delta\sum_{s\in S^*}\mon_1(s)\bar\cval\geq (1-\delta)1+\delta\sum_{s\in S^*}\mon_0(s)\bar\cval.
$$
Let $\Hist_{\mathrm{Pass}}:=\bigcup_{t=0}^\infty (S^*)^t$, the histories in which the incumbent has never generated a signal outside of $S^*$. We will now construct an equilibrium that attains full effort. Let the politician strategy be given by $\pstrat(h)=\mathbf1_{h\in\Hist_{\mathrm{Pass}}}$, the voters' strategies be given by $\vstrat(h)=\mathbf1_{h\notin\Hist_{\mathrm{Pass}}}$, and let the belief map be given by 
$$\belief(h)=\begin{cases}
\prior & \text{ if } h \in \Hist_{\mathrm{Pass}}\\
\prior & \text{ if } h=(\tilde h,s) \text{ for some } \tilde h \in \Hist_{\mathrm{Pass}} \text{ and } s\in S\setminus S^* \\
0  & \text{otherwise.}
\end{cases}$$ 
This strategy profile clearly leads to effort always being chosen until an incumbent is replaced, so we need only show this construction is in fact an equilibrium.\footnote{\label{fn:Markovian-elaboration}The construction is ``nearly'' Markovian: play depends only on beliefs and the most recent signal. In fact, because beliefs do not change on path, the equilibrium is ``simple'' in the sense of \citet{BS93}.  Our construction in the proof of \autoref{lem: bad equilibrium construction} is also nearly Markovian, but because beliefs do change on path it is not simple in their sense. 
}

A incumbent's incentives at histories outside of $\Hist_{\mathrm{Pass}}$ are satisfied because he will be replaced immediately after any signal, and so it is optimal to shirk. Meanwhile, his incentives at histories in $\Hist_{\mathrm{Pass}}$ follow directly from the properties defining $\bar\cval$: his continuation value at such a history is $\bar\cval$, and he prefers this to the value of a one-shot deviation in which he shirks. A voter's incentives at histories inside $\Hist_{\mathrm{Pass}}$ are satisfied because she receives effort with probability $1$ whatever she does, so she prefers not to bear the nonnegative replacement cost. For her incentives at histories $h\notin \Hist_{\mathrm{Pass}}$, her payoff from not replacing is $\belief(h)\leq\prior$, and her payoff from replacing is $1-\cost\geq\prior$. Finally, the Bayesian property is straightforward: $\belief(\emptyset)=\prior$, and at every history $h$ with $\vstrat(h)<1$, we have $\pstrat(h)=1$ and $\belief(h,s)=\belief(h)$ for every $s\in S$.
\end{proof}

The next lemma shows how to construct an equilibrium that does not attain eventual full effort, given \condref{FEI} and a small replacement cost. This equilibrium has every incumbent being eventually replaced and a positive probability of shirking from new incumbents.

\begin{Lemma}\label{lem: bad equilibrium construction}
If $\cost<1-\prior$ and \condref{FEI} holds, then some equilibrium does not attain eventual full effort.
\end{Lemma}
\begin{proof}
We begin with some useful notation and calculations. First, by \cref{lem: cutoff rule for first-best equilibrium}, some $\bar\cval\in(0,1-\ecost)$ and $\lrat\in\left[1,\max_{s\in S}\frac{\mon_0(s)}{\mon_1(s)}\right)$ exist such that, letting $$S^*:=\curlyb{s\in S:\ \tfrac{\mon_0(s)}{\mon_1(s)}\leq\lrat}, \text{ and } \mon_a^*:=\sum_{s\in S^*}\mon_a(s)\in[0,1) \text{ for } a\in\{0,1\},$$ we have
$
\bar\cval=(1-\delta)(1-\ecost)+\delta\mon_1^*\bar\cval\geq (1-\delta)1+\delta\mon_0^*\bar\cval.
$
The previous inequality tells us $\mon_1^*>\mon_0^*$ so that $\tilde\cval:= \bar\cval-\tfrac{1-\delta}{\delta}\tfrac{\ecost}{\mon_1^*-\mon_0^*}<\bar\cval$. Observe also that \begin{eqnarray*}
\delta(\mon_1^*-\mon_0^*)\tilde\cval
&=& \delta(\mon_1^*-\mon_0^*)\bar\cval - (1-\delta)\ecost \\
&=& \brac{(1-\delta)(1-\ecost)+\delta\mon_1^*\bar\cval}- \brac{(1-\delta)1+\delta\mon_0^*\bar\cval} \\
&\geq& 0,
\end{eqnarray*}
so that $\tilde\cval\geq0$. 
Now, defining $\hat\cval:=(1-\delta)(1-\ecost)+\delta\brac{\mon_1^*\bar\cval+(1-\mon_1^*)\tilde\cval}$, note that $\hat\cval$ is a proper weighted average of the three terms $\tilde\cval<\bar\cval<1-\ecost$, and so is strictly higher than $\tilde\cval$. 
Having established that $0\leq\tilde\cval<\hat\cval$, we thus know that $x:=1-\frac{\tilde\cval}{\hat\cval}\in(0,1]$. 
Moreover, 
\begin{eqnarray*}
&& (1-\delta)+\delta\brac{\mon_0^*\bar\cval+(1-\mon_0^*)\tilde\cval} - \hat\cval\\
&=& (1-\delta)\brac{1-(1-\ecost)}+\delta\curlyb{
\brac{\mon_0^*\bar\cval+(1-\mon_0^*)\tilde\cval}
-
\brac{\mon_1^*\bar\cval+(1-\mon_1^*)\tilde\cval}
}\\
&=& (1-\delta)\ecost+\delta(\mon_0^*-\mon_1^*)(\bar\cval-\tilde\cval)\\
&=&0.
\end{eqnarray*}
In summary, given the above calculations and substituting in $\tilde\cval=(1-x)\hat\cval$, we have that 
\begin{eqnarray*}
x &\in& (0,1], \\
\bar\cval&=&(1-\delta)(1-\ecost)+\delta\mon_1^*\bar\cval\\ 
&\geq& (1-\delta)1+\delta\mon_0^*\bar\cval, \text{ and} \\
\hat\cval&=&(1-\delta)(1-\ecost)+\delta\brac{\mon_1^*\bar\cval+(1-\mon_1^*)(1-x)\hat\cval}\\
&=& (1-\delta)+\delta\brac{\mon_0^*\bar\cval+(1-\mon_0^*)(1-x)\hat\cval}.
\end{eqnarray*}
We will use these calculations to construct an equilibrium of the desired form.

To construct such an equilibrium, fix some $a_0\in\left(\frac{\cost}{1-\prior},\ 1\right)$, which exists by hypothesis. We will next describe a strategy profile and belief map that depends on $a_0$ (though we will suppress this dependence in our notation). We will then argue that the profile does not generate EFE, and that for $a_0$ close enough to $1$ it is an equilibrium. 

To describe the strategy profile and belief map, partition histories into three categories.\footnote{Recall that \autoref{fig: bad equilibrium construction} depicts the strategy profile for the special case of binary signals.} Say a history is a \emph{first-regime} history if no signal in $S^*$ has ever (with this incumbent) appeared; say a history is a \emph{second-regime} history if it is not a first-regime history, and every signal since the first $S^*$ signal has been in $S^*$; and say a history is a \emph{third-regime} history otherwise. 
Consider strategies and beliefs $(\pstrat,\vstrat,\belief)$ as follows:
\begin{itemize}
    \item \underline{In the first regime:} Both voter and incumbent mix. Specifically, the voter always chooses $\vstrat(h)=x$. We recursively construct politician behavior and voter beliefs at a history $h$ in this regime as follows. Start by setting $\pstrat(\emptyset)=a_0$ and $\belief(\emptyset)=\prior$.  Then consider any $h=(h_-,s)$ in this regime with $h_-\in\Hist$ and $s\in S$. Given $\pstrat(h_-)>0$, which holds by induction, set $\belief(h)=\bupdate_{\pstrat(h_-)}\paren{\belief(h_-)|s}\in\left(0,\belief(h_-)\right)$, as signal $s$ is strictly bad news (recall $\lrat\geq1$ and $s\notin S^*$). Set $\pstrat(h)\in[a_0,1)$ so that $\belief(h)+\brac{1-\belief(h)}\pstrat(h)=\prior+(1-\prior)a_0-\cost$; doing so is possible because the definition of $a_0$ ensures $\prior+(1-\prior)a_0-\cost>\prior\geq\belief(h)$.
    \item \underline{In the second regime:} The voter retains the incumbent, and the incumbent works. That is, the voter always chooses $\vstrat(h)=0$, and the politician always chooses $\pstrat(h)=1$. Write $h=(h_-,s)$ for some $h_-\in\Hist$ and $s\in S$. As $\pstrat(h_-)>0$ by construction, set $\belief(h)=\bupdate_{\pstrat(h_-)}\paren{\belief(h_-)|s}.$
    \item \underline{In the third regime:} The voter replaces the incumbent, and the incumbent shirks. That is, the voter always chooses $\vstrat(h)=1$, and the politician always chooses $\pstrat(h)=0$. Write $h=(h_-,s)$ for some $h_-\in\Hist$ and $s\in S$. If $h_-$ is a second-regime history, then $\pstrat(h_-)=1$, so set $\belief(h)=\bupdate_{\pstrat(h_-)}\paren{\belief(h_-)|s}=\belief(h_-)$. If $h_-$ is not a second-regime history, then $\vstrat(h_-)=1$ by construction, so we may set $\belief(h)=0$.\footnote{\label{fn:wPBE-use2}Here is the only other place (besides that noted in \autoref{fn:wPBE-use1}) where we use the flexibility of off-path beliefs afforded by \emph{weak} PBE.}
\end{itemize}
Below, we will argue that the constructed $(\pstrat,\vstrat,\belief)$ is an equilibrium for appropriate choice of the parameter $a_0$. Before doing so, we point out that the lemma would follow if it is. Specifically, 
EFE fails because every officeholder shirks with positive probability at his initial history, and almost surely every officeholder is eventually replaced. To confirm the latter point, note that (i) if an officeholder is not replaced at a first-regime history, then there is a probability bounded away from zero (at least $[\prior + (1-\prior)a_0 - \cost]\mon_1^* > \prior\mon_1^*$) of transitioning to the second regime in the next period; (ii) from the second regime, there is a constant positive probability of transitioning to the third regime in the next period; and (iii) the officeholder is replaced at every third-regime history. So the Borel-Cantelli lemma implies almost sure eventual replacement of any officeholder.

We now confirm that $(\pstrat,\vstrat,\belief)$ is an equilibrium for some $a_0$. The Bayesian property is immediate from the construction. Next, consider politician incentives. These are trivial in the third regime: it is optimal for him to shirk since he will be replaced immediately no matter which signal he generates. His incentives to work in the second regime follow directly from the definition of $\bar\cval$ and $S^*$. For the first regime, we can use the fact that 
$$
\hat\cval=(1-\delta)(1-\ecost)+\delta\brac{\mon_1^*\bar\cval+(1-\mon_1^*)(1-x)\hat\cval}\\
= (1-\delta)+\delta\brac{\mon_0^*\bar\cval+(1-\mon_0^*)(1-x)\hat\cval}.
$$
These equalities tell us that his continuation value from any first-regime history is $\hat\cval$, and that he is indifferent between his two effort choices there. Now, we turn to voter incentives. In the second regime, the voter gets payoff $1$ from keeping the incumbent, and so does so optimally. In the first regime, the voter is indifferent because $\belief(h)+\brac{1-\belief(h)}\pstrat(h)=\prior+(1-\prior)a_0-\cost$. 

All that remains is to verify voter optimality in the third regime. Because the voter's outside option is $\prior+(1-\prior)a_0$ and the opportunistic type always shirks in the third regime, it is optimal for the voter to replace at a third-regime history $h$ if and only if 
\begin{equation}
\belief(h)\leq\prior+(1-\prior)a_0-\cost.\label{e:thirdregimeoptimality}
\end{equation}
Toward that inequality, let us show that some $\bar s\in S$ has 
\begin{equation}
    \belief(h)\leq \bupdate_{a_0}(\prior|\bar s).\label{e:thirdregime}
\end{equation}
First, if the previous history was not a second-regime history, then $\belief(h)=0$ and so  
\eqref{e:thirdregime} holds for any $\bar s\in S$. So consider the complementary case, in which $h=(h_2,s_2)$ for some second-regime history $h_2$ and some $s_2\in S$. By definition of the second regime, there exists some first-regime history $h_1$ (which comprises only signals from $S\setminus S^*$), some signal $s_1\in S^*$, and some (possibly empty) string $h_+$ of signals from $S^*$ such that $h_2=(h_1,s_1,h_+)$. Every history $\tilde h$ between\footnote{That is, 
for any history $\tilde h$ that weakly succeeds $(h_1,s_1)$ and weakly precedes $h_2$.} $(h_1,s_1)$ and $h_2$ has $\pstrat(\tilde h)=1$, and so $\belief(\tilde h,\tilde s)=\belief(\tilde h)$ for every signal $\tilde s$, implying  $\belief(h)=\belief(h_2)=\belief(h_1,s_1)$. Since every $s\in S\setminus S^*$ has $\mon_1(s)<\mon_0(s)$, it follows that $0<\bupdate_a(\belief|s)\leq\belief$ for every $\belief\in(0,1]$ and $a\in[0,1]$; hence, by induction $\belief(h_1)\leq\belief(\emptyset)=\prior$. Hence, by construction, the first-regime history $h_1$ has $\pstrat(h_1)\geq\pstrat(\emptyset)=a_0$. Therefore: \begin{itemize}
    \item If $\mon_1(s_1)>\mon_0(s_1)$, we have $\belief(h_1,s_1)=\bupdate_{\pstrat(h_1)}\left(\belief(h_1)\mid s_1\right)
    \leq \bupdate_{a_0}\left(\belief(h_1)\mid s_1\right)\leq \bupdate_{a_0}\left(\prior\mid s_1\right)$. So 
    \eqref{e:thirdregime} holds for $\bar{s}=s_1$.
    \item If $\mon_1(s_1)\leq\mon_0(s_1)$, we have $\belief(h_1,s_1)=\bupdate_{\pstrat(h_1)}\left(\belief(h_1)\mid s_1\right)\leq \belief(h_1)\leq\prior$. So \eqref{e:thirdregime} holds for any $\bar s\in S$ with $\mon_1(\bar s)\geq\mon_0(\bar s)$. 
\end{itemize}

Given that \eqref{e:thirdregime} holds for some $\bar s\in S$, inequality \eqref{e:thirdregimeoptimality} follows if $$\max_{s\in S}\bupdate_{a_0}(\prior| s) \leq \prior+(1-\prior)a_0-\cost.$$
This inequality holds when $a_0$ is close enough to $1$ because as $a_0\to1$, the right-hand side converges to $1-\cost$ while the left-hand side converges to $\prior<1-\cost$.
\end{proof}

\subsubsection{The Equivalence Theorem}

\begin{proof}[Proof of \cref{thm: equivalence theorem}]
Recall that the theorem assumes our maintained assumption in the main text that \mbox{$\cost<\min\{\prior,1-\prior\}$}.

First, suppose \condref{FEI} fails. \cref{prop: differential replacement} then tells us that every equilibrium attains eventual full effort. 
Yet no equilibrium attains full effort because politician incentive-compatibility would be violated (\cref{{lemma: FEI captures effort incentives}}).

Next, suppose \condref{FEI} holds. \cref{lem: first best equilibrium} then implies that some equilibrium attains full effort; and \cref{lem: bad equilibrium construction} demonstrates that not every equilibrium attains eventual full effort.
\end{proof}

\subsection{Other Results}\label{sec: proof extra}

\subsubsection{Replacement Despite a Favorable Reputation}

\begin{proof}[Proof of \cref{prop: firing happens with a good reputation}]
Consider any equilibrium. We want to show it either attains eventual full effort or has $\Prob\curlyb{\beliefrv_\replacetimerv>\prior \text{ and } \replacerv_\infty=1}>0$. We consider three cases.

First, if $\Prob\curlyb{\replacetimerv = \infty}>0$, then the Borel-Cantelli Lemma and symmetry of the equilibrium tell us that, with probability 1, some officeholder is never replaced. But since (by \cref{lem: opportunistic types are replaced}) that officeholder is a good type with conditional probability 1, the equilibrium attains eventual full effort.

Second, if $\Prob\curlyb{\beliefrv_{\replacetimerv}=\prior}=1$, then (because $\mon_1\neq\mon_0$ and beliefs are a martingale) signals are uninformative of effort on path, so the equilibrium attains full effort, hence eventual full effort.

Consider now the remaining case in which  $\replacetimerv<\infty$ a.s.~and $\Prob\curlyb{\beliefrv_\replacetimerv=\prior}<1$. Then, Doob's optional stopping theorem implies $\E[\beliefrv_\replacetimerv]=\prior$. Hence, $\Prob\curlyb{\beliefrv_\replacetimerv>\prior}>0$, delivering the result.
\end{proof}

\subsubsection{Voter Welfare in the Long Run}

\begin{proof}[Proof of \cref{cor: long run voter welfare}]
Suppose \condref{FEI} fails, and that $\cost< \prior$. Fix an equilibrium, and let $B_t$ denote the event that the time-$t$ incumbent is a commitment type who will never be replaced. Observe $B_0\subseteq B_1\subseteq\cdots$ by definition, and \cref{prop: differential replacement} tells us $\Prob\paren{\bigcup_{t=0}^\infty  B_t}=1$. Continuity along chains then implies $\Prob(B_t)\to 1$ as $t\to\infty$. But the time-$t$ voter's ex-ante expected payoff is in $\brac{\Prob(B_t),1}$, delivering the corollary. \end{proof}

\subsubsection{Voters' Outside Option}\label{sec: proof outside option}

The following lemma bounds the growth on incumbent reputation, and hence the growth rate of a voter's value from keeping the incumbent, for a given lower bound on the opportunistic officeholder's expected effort.

\begin{Lemma}\label{lem: belief upper bound is well-behaved}
For any $t\in\Zpos$ and $\pbound\in(0,1)$, the map $\bbound_\pbound^t$ is increasing, and any $\belief\in(0,1)$ has $$\bbound_\pbound^t(\belief)+ \brac{1-\bbound_\pbound^t(\belief)}\pbound\leq \frac{\belief+(1-\belief)\pbound^{t+1}}{\belief+(1-\belief)\pbound^{t}}.$$
Moreover, the right-hand side of the above inequality is increasing in $t$.
\end{Lemma}
\begin{proof}
In the expression defining $\bupdate_a(\belief|s)$, the numerator does not depend on $a$, whereas the denominator is affine (hence monotone) in it, so that $\bupdate_a(\belief|s)$ is monotone in $a$ too. Therefore, 
$$\max_{a\in[\pbound,1]} \bupdate_a(\belief|s)
=\max\curlyb{\bupdate_\pbound(\belief|s), \ \bupdate_1(\belief|s)} = \max\curlyb{\bupdate_\pbound(\belief|s), 
\ \belief}.$$
For each $s\in S$, observe that
$$\bupdate_\pbound(\belief|s)=\frac{\belief}{\belief + (1-\belief)\brac{\pbound + (1-\pbound)\tfrac{\mon_0(s)}{\mon_1(s)}}},
$$
which is strictly decreasing in $\tfrac{\mon_0(s)}{\mon_1(s)}$. Letting $\lrat:=\min_{s\in S}\tfrac{\mon_0(s)}{\mon_1(s)}\in[0,1]$, it follows that 
$$\bbound_\pbound(\belief)=\max_{s\in S} \max\curlyb{\bupdate_\pbound(\belief|s), 
\ \belief}
=
\max\curlyb{\max_{s\in S} \bupdate_\pbound(\belief|s), 
\ \belief}
= 
\frac{\belief}{\belief + (1-\belief)\brac{\pbound + (1-\pbound)\lrat}} \in (0,1).
$$
This quantity is increasing in $\belief$ because $\pbound + (1-\pbound)\lrat\leq1$, and so the composition $\bbound_\pbound^t$ is also increasing.

We now pursue the upper bound. That $\lrat\geq0$ tells us $\bbound_\pbound(\belief)\leq \frac{\belief}{\belief + (1-\belief)\pbound}$, implying $$\frac{\bbound_\pbound(\belief)}{1-\bbound_\pbound(\belief)}\leq \frac{\belief}{(1-\belief)\pbound}.$$
Because $\bbound_\pbound$ is increasing, it follows from induction that $\frac{\bbound^t_\pbound(\belief)}{1-\bbound^t_\pbound(\belief)}\leq \frac{\belief}{(1-\belief)\pbound^t}$ for every $t\in\Zpos$. Rearranging this inequality yields 
$\bbound^t_\pbound(\belief)\leq \frac{\belief}{\belief+(1-\belief)\pbound^t},$
and so (since $\pbound\leq1$)
\begin{eqnarray*}
  \bbound^t_\pbound(\belief) + \brac{1-\bbound^t_\pbound(\belief)}\pbound
  &\leq& \frac{\belief}{\belief+(1-\belief)\pbound^t} + 
\frac{(1-\belief)\pbound^t}{\belief+(1-\belief)\pbound^t}\pbound  \\
&=& \frac{\belief+(1-\belief)\pbound^{t+1}}{\belief+(1-\belief)\pbound^t} \in (0,1).\end{eqnarray*}

All that remains now is to see that this bound increases with $t$. And indeed, the bound is equal to 
$$\frac{\tfrac{\belief}{1-\belief}+\pbound\cdot\pbound^{t}}{\tfrac{\belief}{1-\belief}+\pbound^t},$$
which (because $0<\pbound<1$) is strictly decreasing in $\pbound^t$ and so strictly increasing in $t$.
\end{proof}

The next lemma derives a parametric upper bound on the voters' outside option.
\begin{Lemma}\label{lem: bound on outside option}
Suppose \condref{FEI} fails, and let $T$ be as delivered by \cref{lem: failure of FB is uniform}. In any equilibrium, the voters' outside option is below $$\cost+\inf_{\pbound\in(0,1)}\frac{\prior+(1-\prior)\pbound^{T+1}}{\prior+(1-\prior)\pbound^{T}}.$$
\end{Lemma}
\begin{proof}
Fixing an equilibrium, assume for a contradiction that the voters' outside option $\ooption$ has $\ooption-\cost > \frac{\prior+(1-\prior)\pbound^{T+1}}{\prior+(1-\prior)\pbound^{T}}$ for some $\pbound\in(0,1)$. 

Let $h_0:=\emptyset$, and then let $s_0,\ldots,s_{T-1}\in S$ and $h_1,\ldots,h_{T}\in \Hist$ be as in \cref{lem: continuation value can't keep growing}.

Now, let us show by induction that every $t\in\{0,\ldots,T-1\}$ has both $\pstrat(h_t)\geq\pbound$ and $\belief(h_t)\leq\bbound_\pbound^t(\prior)$. 
For the base case ($t=0$), observe that $\bbound_\pbound^0(\prior)=\prior=\belief(h_0)$, while 
$$\prior+(1-\prior)\pstrat(h_0)=\ooption
>\cost+\frac{\prior+(1-\prior)\pbound^{T}\pbound}{\prior+(1-\prior)\pbound^{T}}
\geq 0+\frac{\prior+(1-\prior)\pbound\cdot 1}{\prior+(1-\prior)1}=\prior+(1-\prior)\pbound\cdot 1,$$
so that $\pstrat(h_0)\geq\pbound$.

Let us turn next to the inductive step. Suppose $t\in\{0,\ldots,T-2\}$ has both $\pstrat(h_t)\geq\pbound$ and $\belief(h_t)\leq\bbound_\pbound^t(\prior)$; we aim to show $t+1$ enjoys the same properties. First, we have 
$$\belief(h_{t+1})\leq\bbound_\pbound\paren{\belief(h_{t})}\leq \bbound_\pbound\paren{\bbound_\pbound^t(\prior)}=\bbound_\pbound^{t+1}(\prior),
$$
where the first inequality holds because $\pstrat(h_t)\geq\pbound$, and the second inequality holds (given \cref{lem: belief upper bound is well-behaved}) because $\belief(h_t)\leq\bbound_\pbound^t(\prior)$. Second, observe that $\pstrat(h_t)\geq\pbound>0$ would not be incentive compatible if $\brac{1-\vstrat(h_t,s)}\cval(h_t,s)$ were zero for every $s\in S$. Hence, given the definition of $s_t$ and $h_{t+1}$, we necessarily have $\vstrat(h_{t+1})<1$. Voter incentives then imply 
\begin{eqnarray*}
\ooption-\cost 
&\leq& \belief(h_{t+1}) + \brac{1-\belief(h_{t+1})}\pstrat(h_{t+1}) \\
&\leq& \bbound_\pbound^{t+1}(\prior) + \brac{1-\bbound_\pbound^{t+1}(\prior)}\pstrat(h_{t+1}).
\end{eqnarray*}
Meanwhile, by hypothesis and given \cref{lem: belief upper bound is well-behaved}, we have
\begin{eqnarray*}
\ooption-\cost 
&>& \frac{\prior+(1-\prior)\pbound^{T+1}}{\prior+(1-\prior)\pbound^{T}}\\
&\geq&\bbound_\pbound^{t+1}(\prior) + \brac{1-\bbound_\pbound^{t+1}(\prior)}\pbound.
\end{eqnarray*}
Combining the two inequality chains yields $\pstrat(h_{t+1})\geq\pbound$, completing the inductive step.

That $\pstrat(h_t)>0$ for every $t\in\{0,\ldots,T-1\}$ contradicts \cref{lem: continuation value can't keep growing}.
\end{proof}

To prove \autoref{prop: outside option can be super low}, we show that the bound of \cref{lem: bound on outside option} can be made arbitrarily small if the prior and replacement costs are both small.\footnote{Although \cref{sec: model} assumes $\cost<\min\{\prior,1-\prior\}$, the restriction is evidently not needed for \cref{prop: outside option can be super low}: if \condref{FEI} fails and $\cost$ and $\prior$ are both very small, then the outside option is very small in all equilibria, irrespective of the relationship between $\cost$ and $\prior$.}

\begin{proof}[Proof of \cref{prop: outside option can be super low}]
Let $T$ be as given by \cref{lem: failure of FB is uniform}. 
Because \begin{eqnarray*}
0&\leq&\cost+\inf_{\pbound\in(0,1)}\frac{\prior+(1-\prior)\pbound^{T+1}}{\prior+(1-\prior)\pbound^{T}}\\
&\leq& \cost+\frac{\prior+(1-\prior)\prior^{(T+1)/(T+2)}}{\prior+(1-\prior)\prior^{T/(T+2)}}\\
&\to& 0 \text{ as } (\cost,\prior)\to(0,0),
\end{eqnarray*}
it follows that $\cost+\inf_{\pbound\in(0,1)}\frac{\prior+(1-\prior)\pbound^{T+1}}{\prior+(1-\prior)\pbound^{T}}$ converges to zero as $\cost$ and $\prior$ do. The proposition then follows from \cref{lem: bound on outside option}.
\end{proof}

Finally, although not stated in the main text, we record one further result that extends our main equivalence result to the voters' outside option being bounded away from first best, so long as the replacement cost is sufficiently small.

\begin{Corollary}\label{cor: equivalence theorem with extra condition}
For any $(\ecost,\delta,\mon,\prior)$, there exists $\bar\cost>0$ such that, if $\cost<\bar\cost$, then the following are equivalent:

\begin{enumerate}
  \item \label{cor, allEFE} All equilibria attain eventual full effort;
  \item \label{cor, noneFB} No equilibrium attains full effort;
  \item \label{cor, FBIfails} \condref{FEI} fails.
  \item \label{cor, OO} Across all equilibria, the voters' outside option is bounded away from $1$.
\end{enumerate}
\end{Corollary}
\begin{proof}
Let $\hat\cost:=\min\{\prior,1-\prior\}$. If $(\ecost,\delta,\mon)$ satisfy \condref{FEI}, then let $\bar\cost:=\hat\cost$. If $(\ecost,\delta,\mon)$ violate \condref{FEI}, then let $\bar\cost\in(0,\hat\cost)$ be such that $\bar\cost+\inf_{\pbound\in(0,1)}\frac{\prior+(1-\prior)\pbound^{T+1}}{\prior+(1-\prior)\pbound^{T}}<1$, where $T$ is as given by \cref{lem: failure of FB is uniform}. Note that such $\bar\cost$ exists because $\frac{\prior+(1-\prior)\pbound^{T+1}}{\prior+(1-\prior)\pbound^{T}}<1$ for any given $\pbound\in(0,1)$.

\cref{thm: equivalence theorem} tells us conditions \ref{cor, allEFE}, \ref{cor, noneFB}, and \ref{cor, FBIfails} are all equivalent; and obviously condition \ref{cor, OO} implies \ref{cor, noneFB}. Lastly, \cref{lem: bound on outside option} tells us condition \ref{cor, FBIfails} implies \ref{cor, OO}.
\end{proof}

\afterpage{
    \clearpage 
    \begin{table}[htbp]
        \centering
        \makebox[\textwidth][c]{
        \begin{threeparttable}[htbp]
            \centering
            \renewcommand{\arraystretch}{1.5}
            \begin{tabular}{l p{13.5cm} @{\extracolsep{2em}} l}
                \toprule
                \textbf{Symbol} & \textbf{Description} & \textbf{Defined In} \\
                
                \midrule
                \multicolumn{3}{l}{\textit{Histories, Strategies, and Values}} \\
                $\Hist$ & Set of all (public) career histories & \secref{sec: model} \\
                $\mon_{\aprob}(s)$ & Probability of signal $s$ given $\aprob$ expected effort: equals $\aprob\mon_1(s)+(1-\aprob)\mon_0(s)$ & \secref{sec: model} \\
               $\pstrat(h)$ & Probability an opportunistic incumbent exerts effort at history $h$ & \secref{sec: model} \\
                $\vstrat(h)$ & Probability the voter replaces the incumbent at history $h$ & \secref{sec: model} \\
                $\belief(h)$ & Incumbent's reputation (belief that he is the {good} type) at history $h$ & \secref{sec: model} \\
                $\eeff(h)$ & Expected effort at history $h$: equals $\belief(h) + \brac{1-\belief(h)}\pstrat(h)$ & \secref{sec: model} \\
                $\cval(h)$ & Opportunistic incumbent's continuation value at history $h$ & \secref{sec: model} \\
                $\ooption$ & Voters' outside option: equals $\prior+(1-\prior)\pstrat(\emptyset)$ & \secref{subsec:good-longrun} \\    
                \midrule
                \multicolumn{3}{l}{\textit{Belief Operators}} \\
                $\bupdate_a(\belief \mid s)$ & Bayesian posterior from belief $\belief$, opportunistic type's effort $a$, and signal $s$ & \secref{sec: proof preliminaries} \\
                $\bbound_\eta(\belief)$ & Highest possible Bayesian posterior from belief $\belief$ and opportunistic type's effort at least $\eta$ & \secref{sec: proof preliminaries} \\
                $\bbound_\pbound^t$ & The $t$-fold composition of the operator $\bbound_\pbound$ & \secref{sec: proof preliminaries} \\
                
                \midrule
                \multicolumn{3}{l}{\textit{Random Variables}} \\
                $\replacerv_t$ & Indicator variable for replacement: equals 1 if first\tnote{a} \ officeholder has been replaced by time $t$ & \secref{sec: proof EFE}\\
                $\replacetimerv$ & Replacement time of first officeholder: equals $\inf\{t : \replacerv_t = 1\}\in\Zpos\cup\{\infty\}$ & \secref{sec: proof EFE} \\
                $\beliefrv_t$ & First officeholder's reputation at time $t$ (stopped at $\replacetimerv$) & \secref{sec: proof EFE} \\
                $\actprv_t$ & Officeholder's expected effort at time $t$ conditional on being opportunistic & \secref{sec: proof EFE} \\
                $\beliefrv_\infty, \replacerv_\infty, \actprv_\infty$ & Almost-sure limits of the stochastic processes $\beliefrv_t$, $\replacerv_t$, $\actprv_t$ & \secref{sec: proof EFE} \\
                \bottomrule
            \end{tabular}
        
        \begin{tablenotes}
        \footnotesize
        \item[a] By equilibrium symmetry, an analog applies to other officeholders for this and other relevant random variables
        \end{tablenotes}

            \caption{Summary of Key Notation used in the Proof Appendix}
        \label{table:notation}    
        \end{threeparttable}
        }
    \end{table}
    \clearpage 
}

\newpage
\bibliographystyle{ecta}
\bibliography{replacement.bib}
\end{document}